\documentclass[twoside,11pt]{article}
\usepackage{amsmath, amsthm, amsfonts}
\usepackage[figuresright]{rotating}
\usepackage[utf8]{inputenc}
\usepackage[english]{babel}
\usepackage[ruled, vlined]{algorithm2e}
\usepackage[colorlinks,citecolor=blue,urlcolor=blue]{hyperref}

\def \v{\varepsilon}

\def \b{\mathbf}
\def \bs{\boldsymbol}
\def \a{\alpha}
\newcommand{\y}{\b{y}}
\newcommand{\X}{\b{X}}
\newcommand{\x}{\b{x}}
\newcommand{\A}{\mathcal{A}}
\newcommand{\M}{\mathcal{M}}
\newcommand{\Var}{\text{Var}}
\newcommand{\MCV}{\text{MCV}}
\newcommand{\MF}{\text{MF}}
\newcommand{\WMF}{\text{WMF}}

\numberwithin{equation}{section}
\theoremstyle{plain}
\newtheorem{thm}{Theorem}
\newtheorem{remark}{Remark}
\newtheorem{lemma}{Lemma}
\newtheorem{corollary}{Corollary}

\begin{document}
\renewcommand{\baselinestretch}{1}
\def\Quote{\begin{quotation}\normalfont\small}
\def\EndQuote{\end{quotation}\rm}
\def\BigHeading{\bfseries\Large}\def\MediumHeading{\bfseries\large}
\def\bct{\begin{center}}
\def\ect{\end{center}}
\font\BigCaps=cmcsc9 scaled \magstep 1
\font\BigSlant=cmsl10    scaled \magstep 1
\def\lbk{\linebreak}
\def\Report{Prediction Weighted Maximum Frequency Selection}
\def\Author{Liu and Rao}
\pagestyle{myheadings}
\markboth{\Author}{\Report}
\thispagestyle{empty}
\bct{\BigHeading
Prediction Weighted Maximum Frequency Selection}\\\vskip15pt
Hongmei Liu and J. Sunil Rao
\vskip5pt
{\small \it Division of Biostatistics, University of Miami, Miami, FL 33136\\
h.liu7@med.miami.edu\\
JRao@biostat.med.miami.edu}
\ect

\Quote
\vskip-5pt\noindent
Shrinkage estimators that possess the ability to produce sparse solutions have become increasingly important to the analysis of today's complex datasets.  Examples include the LASSO, the Elastic-Net and their adaptive counterparts.  Estimation of penalty parameters still presents difficulties however.  While variable selection consistent procedures have been developed, their finite sample performance can often be less than satisfactory.   We develop a new strategy for variable selection using the adaptive LASSO and adaptive Elastic-Net estimators with $p_n$ diverging.   The basic idea first involves using the trace paths of their LARS solutions to bootstrap estimates of maximum frequency (MF) models conditioned on dimension.  Conditioning on dimension effectively mitigates overfitting, however to deal with underfitting, these MFs are then prediction-weighted, and it is shown that not only can consistent model selection be achieved, but that attractive convergence rates can as well, leading to excellent finite sample performance.  Detailed numerical studies are carried out on both simulated and real datasets.  Extensions to the class of generalized linear models are also detailed. 
\EndQuote

\section{Introduction}
\label{sec:intro}

Consider the standard linear regression model
\begin{eqnarray}
\b{y}=\b{X}\bs{\beta}+\bs{\v}, 
\label{eqt1}
\end{eqnarray}
where $\b{y} = (y_1, \dots, y_n)^T$ is a vector of responses, $\b{X} = (\X_1, \dots, \X_{p_n})$ is an $n \times p_n$ design matrix of predictors, ${\bs \beta}=(\beta_1, \dots, \beta_{p_n})^T$ is a vector of unknown regression parameters, $\bs \v = (\v_1, \dots, \v_n)^T$ is a vector of independent and identically distributed (i.i.d.) random errors. We allow $p_n$ to increase with $n$.

Because some elements of ${\bs \beta}$ might be $0$, a family of penalized least squares estimators were developed for variable selection and estimation,   
\begin{eqnarray}
\bs{\hat{\beta}} = \arg \min_{\bs \beta} \|\b{y} - \b{X}\bs{\beta}\|^2 + \sum_{j=1}^{p_n}\rho(|\beta_j|, \bs \lambda), 
\label{eqt2}
\end{eqnarray}
where $\|\cdot\|$ is the $L_2$-norm, $\bs \lambda \ge \b{0} $ are regularization parameters, and $\rho(|\beta_j|, \bs{\lambda})$ is positive valued for $\beta_j \ne 0$.  \cite{fan2006statistical} pointed out that through variable selection one can focus on a small number of important predictors for enhanced scientific discovery and potentially improve prediction performance by removing noise variables.

Penalized $L_q$-regression is a special case of (\ref{eqt2}) with
$\rho(|\beta_j|, \bs \lambda) = \lambda|\beta_j|^q, q \ge 0$, which includes the best subset selection for $q=0$; the LASSO \cite{tibshirani1996regression} for $q=1$ and the ridge regression \cite{hoerl1970ridge} for $q=2$.  Best subset selection is known to be computationally infeasible for high dimensional data and inherently discrete in variable selection \cite{breiman1995better}.  Ridge regression does not possess a variable selection property. The LASSO however, can do simultaneous estimation and variable selection because its $L_1$ penalty is singular at the origin and can shrink some coefficients to exact 0 with a sufficiently large $\lambda$ \cite{fan2001variable}.  Other penalized least squares estimators that can do simultaneous estimation and variable selection include the SCAD \cite{fan2001variable} and adaptive LASSO \cite{zou2006adaptive} both enjoying the oracle properties \cite{fan2001variable}; the Elastic-Net \cite{zou2005regularization} capable of detecting grouped effects; the adaptive Elastic-Net \cite{zou2009adaptive} combining advantages of the adaptive LASSO and Elastic-Net; and etc.

Selection of $\bs{\lambda}$ is essential in above penalized least squares estimation procedures.  Although methods such as the SCAD, adaptive LASSO and adaptive Elastic-Net enjoy the oracle properties asymptotically, their optimal properties rely on particular specifications of the $\bs{\lambda}$, whose magnitude controls the complexity of a selected model and trade-off between bias and variance in estimators \cite{fan2010selective}.  The multi-fold cross-validation (CV) and generalized cross-validation (GCV) are frequently applied for the tuning parameters selection  \cite{tibshirani1996regression,fan2001variable,zou2006adaptive,zou2005regularization}.  But they overfit the model asymptotically \cite{zhang1993model}.  For consistent variable selection, \cite{wang2007unified} suggested to use the BIC in adaptive LASSO and a modified BIC when $p_n$ is diverging \cite{wang2009shrinkage}; \cite{chen2008extended} introduced an extended BIC (EBIC) for linear models and then generalized it to generalized linear models \cite{chen2012extended}; \cite{fan2013tuning} put forward a generalized information criterion (GIC) with $p_n$ diverging; \cite{feng2013consistent} provided a consistent cross-validation procedure (CCV) for the LASSO; \cite{meinshausen2010stability} proposed the stability selection (SS) for their randomized LASSO.  Although variable selection consistency was established for these procedures,  their finite sample performance can often be less than optimal (Section~\ref{sec:simu} ahead demonstrates this in simulation studies).  

\begin{figure}
\centering
\includegraphics[scale=0.38]{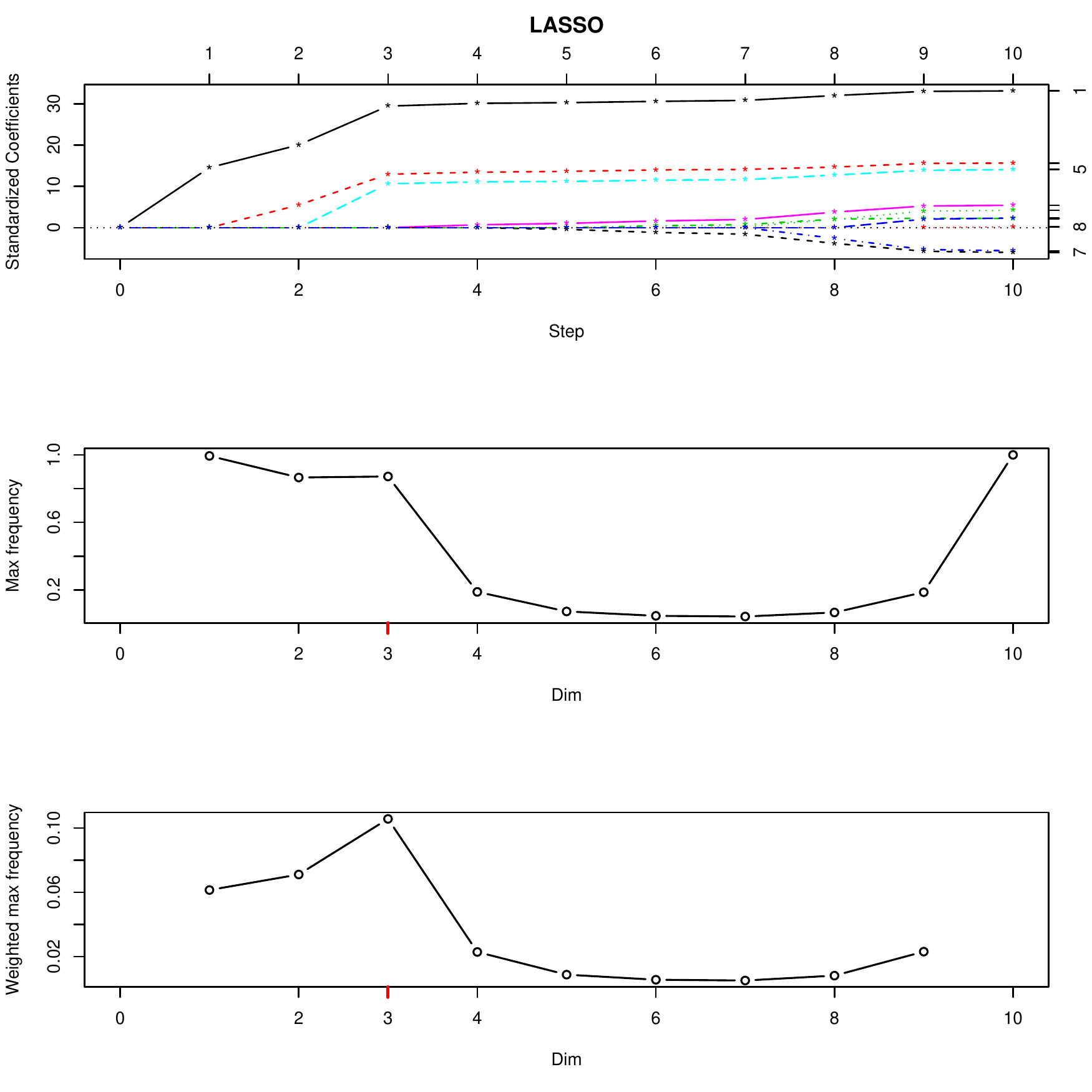}
\caption{The full adaptive LASSO solution path from the LARS (top panel), the estimated maximum frequency at each dimension (middle panel) and the weighted maximum frequency at each dimension (bottom panel). The red tick in $x$-axis indicates the true dimension 3. 
\label{fig:intro1}}
\end{figure}

We propose a new method for tuning parameters selection, focusing in particular on the adaptive LASSO and adaptive Elastic-Net estimators.  A simple example helps to illustrate the basic idea.  Consider the adaptive LASSO in following example.

{\it Example 1.} Data are drawn from model (\ref{eqt1}) with $\bs{\beta} = (3, 1.5, 0, 0, 2, 0, \dots)_{10}^T$,  row vectors of the design matrix $\b{x}_i \overset{iid}{\sim} N_{10}(\b{0}, \bs{\Sigma})$ with $\bs{\Sigma}(i, j)=0.3^{|i-j|}$ and $\v_i \overset{iid}{\sim} N(0, 3^2)$ for $i=1, \dots, 100$.  So the true model size here is 3.

Figure~\ref{fig:intro1} (top) shows the full adaptive LASSO solution path from the LARS algorithm \cite{efron2004least}. In the figure, each step indicates a dimension change in the estimator. These steps are called transition points in \cite{zou2007degrees}. They showed that if using information criteria such as the AIC or BIC to identify the optimal $\lambda$ in adaptive LASSO, it lies in one of the transition points. This result helps to justify uses of the LARS algorithm and our subsequent focus on the transition points.  Then the question remains about how to choose from these transition points.

Figure~\ref{fig:intro1} (middle and bottom) gives a brief look at our proposed method.  The 
middle panel shows the estimated maximum frequency (MF) of a candidate model given the dimension.  The MF estimation is done by a bootstrapping algorithm using the transition points. The strategy of conditioning on dimension has two important consequences:  i) for overfit dimensions, the MFs are dramatically smaller than the true dimension MF (other than the full model of course),  and ii) underfit dimensions can also produce large MF values.  Point i) is important because for variable selection, overfitting is usually much more difficult to deal with.  So one must now deal with the underfitting issue.  We do this by introducing a prediction-based weight to the MFs (labeled as WMF). The results are shown in bottom panel of Figure~\ref{fig:intro1}.  As is evident, now the true dimension, which maps to the true model, stands out beautifully from all others.  





The rest of the paper is organized as follows. In Section~\ref{sec:RB}, we briefly review the adaptive LASSO and adaptive Elastic-Net estimators and introduce the bootstrap algorithm for each.  In Section~\ref{sec:MF}, the MF procedure itself is described and its underlying properties are carefully examined using a simple orthogonal design.  In Section~\ref{sec:asMF}, asymptotic properties of the MF procedure are established in general settings.  The WMF procedure and its variable selection consistency are presented in Section~\ref{sec:WMF}.  Comprehensive simulation studies are shown in Section~\ref{sec:simu}.  Section~\ref{sec:ext} describes extensions of the MWF procedure to generalized linear models (GLMs).  Applications of the WMF procedure to ultra-high dimensional data are discussed in Section \ref{sec:uhd}.

\section{Bootstrapping the adaptive LASSO and adaptive Elastic-Net estimators}
\label{sec:RB}

Denote $\bs{\beta}_0$ the true value of $\bs{\beta}$ with model size $p_0$, and $\bs{\tilde{\beta}} = (\tilde{\beta}_1, \dots, \tilde{\beta}_{p_n})^T $ a consistent estimate of $\bs{\beta}_0$. The adaptive LASSO \cite{zou2006adaptive} estimator is 
\begin{eqnarray}
\bs{\hat{\beta}}_a = \arg \min_{\bs \beta} ||\b y - \b X \bs \beta||^2 + 2\lambda_n \sum_{j=1}^{p_n} \omega_j|\beta_j|,
\label{eqt3}
\end{eqnarray}
where $\omega_j = |\tilde{\beta}_j|^{-\gamma}$, $\gamma \ge 0$.  It was suggested to use the ordinary least-squares (OLS) estimator or the best ridge estimator (if collinearity exists) for $\bs{\tilde{\beta}}$.  Under certain regularity conditions, $\bs{\hat{\beta}}_a$ was shown to enjoy the oracle properties.

The Elastic-Net estimator \cite{zou2005regularization} is  
\begin{eqnarray}
\bs{\hat{\beta}}_e = (1+{\lambda_{n2} \over n} ) \left \{ \arg \min_{\bs \beta} ||{\b y}-{\b X}{\bs \beta}||^2+\lambda_{n2}\sum_{j=1}^{p_n}|\beta_j|^2+\lambda_{n1}\sum_{j=1}^{p_n}|\beta_j| \right \}. 
\label{eqt4}
\end{eqnarray} 
It overcomes several limitations pertaining to the LASSO: 1) the added $L_2$ penalty is {\it strictly} convex to allow {\it grouping effects}; 2) In a $p_n>n$ case, it can potentially estimate all $p_n$ predictors, while the LASSO can only find at most $n$ predictors.

\cite{zou2009adaptive} proposed the adaptive Elastic-Net to combine strengths of the Elastic-Net and adaptive LASSO. The adaptive Elastic-Net estimator is 
\begin{eqnarray}
\bs{\hat{\beta}}_{ae} = (1+{\lambda_{n2} \over n}) \left \{\arg \min_{\bs \beta} ||{\b y}-{\b X}{\bs \beta}||^2+\lambda_{n2}\sum_{j=1}^{p_n}|\beta_j|^2+\lambda_{n1}^+\sum_{j=1}^{p_n} \omega_j|\beta_j|  \right \},
\label{eqt5}
\end{eqnarray}
where $\omega_j = |\hat{\beta}_{ej}|^{-\gamma}$, $\gamma \ge 0$ and $\bs{\hat{\beta}}_e=(\hat{\beta}_{e1}, \dots, \hat{\beta}_{ep_n})^T$ is the Elastic-Net estimator in (\ref{eqt4}).  Note that $\lambda_{n2}$ takes the same value for the $L_2$ penalty function in (\ref{eqt4}) and (\ref{eqt5}), because the $L_2$ penalty contributes to the same kind of grouping effects.  On the other hand, $\lambda_{n1}$ and $\lambda_{n1}^+$ are allowed to be different as they control the sparsity in estimators.  Under some regularity conditions, $\bs{\hat{\beta}}_{ae}$ was shown to enjoy the oracle properties.  

We now detail bootstrapping for these two estimators.  There are typically two ways of generating bootstrap observations for model (\ref{eqt1}) \cite{shao1996bootstrap}.  

1. {\it Bootstrapping pairs} \cite{efron1982jackknife}.  Let $\hat{F}(\X, \b{y})$ be the empirical distribution putting mass $n^{-1}$ on each data pair $(\x_i, y_i), i=1, \dots, n$.  Generate i.i.d. paired bootstrap data $\{(\x_i^*, y_i^*), i = 1, \dots, n\}$ from $\hat{F}(\X, \b{y})$.  The bootstrap analog of $\bs{\hat{\beta}}_a$, denoted as $\bs{\hat{\beta}}_a^*$, is to replace $(\X, \b{y})$ with $(\X^*, \b{y}^*)$ in (\ref{eqt3}) where $\X^* = (\x_1^*, \dots, \x_n^*)^T$ and $\b{y}^* = (y_1^*, \dots, y_n^*)^T$.  So is the bootstrap analog of $\bs{\hat{\beta}}_{ae}$, denoted as $\bs{\hat{\beta}}_{ae}^*$.  Under the weak condition that $\X^T\X \to \infty$,   $\X^T\X(\X^{*T}\X^*)^{-1} \to 1$ almost surely \cite{shao1996bootstrap}.

2. {\it Bootstrapping residuals} \cite{efron1979computers}.  Calculate the $i$th residual 
$$\hat{\v}_{0i} = y_i - \x_i^T\bs{\hat{\beta}},$$
where $\bs{\hat{\beta}}$ is a ridge estimate of $\bs{\beta}_0$.  Generate i.i.d. bootstrap residuals $\{\v_i^*, i=1, \dots, n\}$ from the empirical distribution that puts mass $n^{-1}$ on each centered residual, $\hat{\v}_i = \hat{\v}_{0i} - \bar{\v}_0$, where $\bar{\v}_0$ is the average of $\hat{\v}_{0i}, i = 1, \dots, n$.  Then the i.i.d. residual bootstrap data is $\{(\x_i, y_i^*), i=1, \dots, n\}$ where $y_i^* = \x_i^T\bs{\hat{\beta}} + \v_i^*$.  The bootstrap analog of $\bs{\hat{\beta}}_a$, denoted as $\bs{\tilde{\beta}}_a^*$, is to substitute $\b{y}$ with $\b{y}^*$ in (\ref{eqt3}).  So is the bootstrap analog of $\bs{\hat{\beta}}_{ae}$, denoted as $\bs{\tilde{\beta}}_{ae}^*$.

In next section, we introduce the MF procedure which takes use of above bootstrap estimators.

\section {The MF procedure}
\label{sec:MF}

Denote a $j$-dimensional candidate model from the $i$th bootstrap data as $M_j^i$. 

\begin{algorithm}[H]
1. Draw $B$ (residual or paired) bootstrap data \;
2. Use the LARS algorithm to fit each bootstrap data, then get $B$ collections of candidate models, $\{M_1^i, \dots, M_{p_n}^i\},  i=1, \dots, B$  \;
3. At each dimension $j$, count the frequency of each unique model in $\{M_j^1, \dots, M_j^B\}$, denoted as $\{c_{j1}, \dots, c_{jt}\}$ where $t$ is the number of unique models.  Let $\MF_j=\max \{c_{j1}, \dots, c_{jt}\}$ corresponding to model $M_j$ \;
4. Select the dimension $r^*$ and model $M_{r^*}$ s.t.
$$ r^* = \max \{j: \ j = \arg \max_{1 \le i \le p-1} \MF_i \}. $$
\caption{The MF procedure for adaptive LASSO}
\label{alg1}
\end{algorithm}

\begin{remark}
In the 4th step, the full model is excluded because it will destroy the maximum frequency rule by having the highest frequency, $B$, all the time.  If there is a tie at the maximum of $\MF_i, 1 \le i \le p-1$, we select the one at the highest dimension.  This strategy guarantees asymptotic variable selection consistency of the MF procedure, which will be discussed in Section~\ref{sec:asMF}. 
\end{remark}

The MF procedure for adaptive Elastic-Net is in parallel.  But in the 2nd step, the LARS-EN algorithm \cite{zou2005regularization} is used instead to fit each bootstrap data.

We discussed in introduction to this paper consequences of the MF procedure by conditioning on dimension.  Here we use a simple orthogonal design with i.i.d. normal random errors to study underlying properties driving that performance.  In this case, we have $\b X^T \b X=\b I$ and the adaptive Elastic-Net reduces automatically to the adaptive LASSO \cite{zou2009adaptive}.  Denote $\X_j$ the $j$th column of $\X$.   Then the adaptive LASSO estimator is
\begin{eqnarray}
\hat{\beta}_j=\{|\X_j^{T}{\b y}| - \frac{\lambda_n}{|\tilde{\beta}_j|^\gamma}\}_{+}sgn(\X_j^{T} \b y), \quad j = 1, \dots, p_n, 
\label{eqt15}
\end{eqnarray}
where $z_+$ equals to $z$ if $z>0$ otherwise 0.  We can expand the $\X_j^{T} \b y$ by 
$$ \X_j^{T} \b y=\beta_{0j} + \X_j^T{\bs \v}, $$
where $\X_j^T{\bs \v} \sim N(0, \sigma^2)$.  The following Lemma gives an order relationship for $\X_j^{T} \b y$'s.

\begin{lemma}
Suppose ${\b X}^T{\b X}={\b I}$, then we have 
\begin{align}
P \left (|{\X}_i^{T} \b y| > |{\X}_j^{T} \b y| \right)>0.5 \quad \text{if} \quad |\beta_{0i}|>|\beta_{0j}|, \nonumber \\ 
P \left (|{\X}_i^{T} \b y| > |{\X}_j^{T} \b y| \right)=0.5 \quad \text{if} \quad |\beta_{0i}|=|\beta_{0j}|, \nonumber
\end{align}
for $ i, j \in \{1, \dots, p_n\} $.
\label{lem1} 
\end{lemma}

In combine with the fact that $\frac{\lambda_n}{|\tilde{\beta}_i|^\gamma} > \frac{\lambda_n}{|\tilde{\beta}_j|^\gamma}$ asymptotically for $\beta_{0i} < \beta_{0j}$, it is easy to deduce from (\ref{eqt15}) that given a $\lambda_n$ adaptive LASSO tends to select those variables, corresponding to the first $k_{\lambda_n}$ largest $|\beta_j|$'s, with the highest probability.

Without loss of generality, suppose $|\bs{\beta}_0|$ is decreasingly ordered.  Denote $\mathcal{S}_r$ a $r$-dimensional model containing the first $r$ elements of $|\bs{\beta}_0|$,  and denote $\mathcal{W}_r$ any other $r$-dimensional models.  Let $\hat{\A}_r$ be an adaptive LASSO model estimate given the model size is $r$, $P(\hat{\A}_r = \mathcal{S}_r \mid r)$ indicates the conditional probability of $\hat{\A}_r = \mathcal{S}_r$ given the model size.  Then preceding deductions from (\ref{eqt15}) can be formularized as  
\begin{align}
&(1). \quad P(\hat{\A}_r = \mathcal{S}_r \mid r) > P(\hat{\A}_r = \mathcal{W}_r \mid r), \ 0 < r \le p_0,
\label{eqt17} \\
&(2). \quad  P(\hat{\A}_r = \mathcal{W}_{r}^1 \mid r) = P(\hat{\A}_r = \mathcal{W}_{r}^2 \mid r), \ p_0 < r < p_n,
\label{eqt16}
\end{align}
where $\mathcal{W}_{r}^1$ and $\mathcal{W}_{r}^2$ are two $r$-dimensional models s.t.
$\mathcal{S}_{p_0} \subset \mathcal{W}_{r}^1, \mathcal{W}_{r}^2$. 

Above properties of the adaptive LASSO coincides to some extent with the results of Theorem 2 in \cite{dey2008depth}.  By (\ref{eqt16}), zero predictors will be equally likely selected at an overfit dimension.  As a result $P(\hat{\A}_r = M_r \mid r)$ (see Algorithm 1 for definition of $M_r$), $p_0 < r < p_n$, drops down dramatically, which is why we see a huge gap between the true dimension and overfit dimensions in Figure~\ref{fig:intro1} (middle).  On the other hand, $P(\hat{\A}_r = \mathcal{S}_r \mid r)$ at some underfit dimensions can be as competitive as $P(\hat{\A}_r = \mathcal{S}_{p_0} \mid p_0)$.  We propose a WMF procedure to tackle this underfitting issue in Section \ref{sec:WMF}.

In next section, we show asymptotic variable selection properties for $\bs{\hat{\beta}}_a^*$ and $\bs{\tilde{\beta}}_{ae}^*$ in general settings, from which variable selection consistency of the MF procedure can be deduced.

\section{Asymptotic properties of the MF procedure}
\label{sec:asMF}

Let $\A = \{j: \beta_{0j} \neq 0 \}$ be the true model.  We assume following regularity conditions for subsequent theoretical studies:

(A1)  Denote $\zeta_{min}(\b{C})$ and $\zeta_{max}(\b{C})$ the minimum and maximum eigenvalues of a positive definite matrix $\b{C}$.  We assume
$$  d \le \zeta_{min}({1 \over n} \b X^T \b X) \le \zeta_{max}({1 \over n} \b X^T \b X) \le D, $$
where $d$ and $D$ are two positive constants. 

(A2) $p_n = n^\varrho, \ 0 \le \varrho < 1$ and $\gamma > {\varrho \over 1-\varrho}$.  The last inequation is to ensure $(1-\varrho)(1+\gamma)>1$ in (A3)--(A4).  Moreover,
$$\lim_{n \to \infty} {p_n \over n} {1 \over \min_{j \in \A}|\beta_{0j}|^2} \to 0.$$

(A3)  In adaptive LASSO, 
$$\lim_{n \to \infty} \lambda_n / \sqrt{n} \to 0, \quad \lim_{n \to \infty} {\lambda_n \over \sqrt{n}} n^{{(1-\varrho)(1+\gamma)-1 \over 2}} \to \infty,$$
and 
$$\lim_{n \to \infty} \left({\lambda_n \over \sqrt{n}}\right)^{1 \over \gamma} {1 \over \min_{j \in \A}|\beta_{0j}|} \to 0.$$

(A4) In adaptive Elastic-Net, 
$$ \lim_{n \to \infty} \lambda_{n1} / \sqrt{n} \to 0, \quad \lim_{n \to \infty} \lambda_{n2} / \sqrt{n} \to 0, $$
and
$$\lim_{n \to \infty} \lambda_{n1}^+ / \sqrt{n} \to 0, \quad \lim_{n \to \infty} {\lambda_{n1}^+ \over \sqrt{n}} n^{{(1-\varrho)(1+\gamma)-1 \over 2}} \to \infty,$$ 
$$ \lim_{n \to \infty} \left({\lambda_{n_1}^+ \over \sqrt{n}}\right)^{1 \over \gamma} {1 \over \min_{j \in \A}|\beta_{0j}|} \to 0.$$

(A5) The errors $\{\v_i, i = 1, \dots, n\}$ are i.i.d. with mean 0 and variance $\sigma^2 < \infty$.
\vspace{2mm}


Denote $\A_n^* = \{j: \hat{\beta}_{aj}^* \neq 0 \}$ an adaptive LASSO estimate of $\A$ using paired bootstrap data.  Let $P^* = P(\cdot \mid \mathcal{E})$ and $E^* = E(\cdot \mid \mathcal{E})$ where $\mathcal{E} = \sigma\left((\x_i, y_i), i=1, \dots, n\right)$.  Then $P^*(\A_n^* = \A \mid \lambda_n)$ indicates the conditional probability of $\A_n^* = \A$ given $\mathcal{E}$ and $\lambda_n$. 

\begin{thm} 
Suppose conditions (A1)--(A3) and (A5) hold, then  
\begin{align*}
\lim_{n \to \infty} P^*(\mathcal{A}_n^* = \mathcal{A} \mid \lambda_n) = 1. 
\end{align*}
Moreover,  let $\lambda_n'$ be another tuning parameter such that the adaptive LASSO estimator under $\lambda_n'$ is of dimension $r$, $p_0<r<p_n$, then 
\begin{align*}
\lim_{n \to \infty} P^*(\mathcal{A}_n^* = \mathcal{M}_{r} \mid \lambda_n') < 1,
\end{align*}
where $\mathcal{M}_{r}$ is any $r$-dimensional model. 
\label{thm1}
\end{thm}

Proofs of Theorem \ref{thm1} are included in Appendix \ref{AP}. 

In adaptive LASSO, given a $\lambda_n$ is equivalent to given a dimension, but the converse is not true.  One dimension can be mapped to numerous models, as a result to numerous tuning parameters.  Fortunately however, the LARS algorithm enables us to map a dimension to an optimal $\lambda_n$.  Recall the adaptive LASSO solution path from the LARS in top panel of Figure \ref{fig:intro1}.  Transition points (e.g. steps) from 0 to 10 corresponds to a sequence of $\lambda_n$'s: 
$$ \lambda_n(0) > \lambda_n(1) > \dots > \lambda_n(10)=0. $$
Note that $\bs{\hat{\beta}}_a(\lambda_n) = \b{0}$ for $\lambda_n > \lambda_n(0)$ where $\bs{\hat{\beta}}_a(\lambda_n)$ is the adaptive LASSO estimator under $\lambda_n$.  By Theorem 5 in \cite{zou2007degrees}, 
$$ \lambda_n(m+1) = \arg\min_{\lambda_n} \|\b{y} - \X\bs{\hat{\beta}}_a(\lambda_n)\|^2 + a_n \hat{df}(\lambda_n),\ \lambda_n(m+1) \le \lambda_n < \lambda_n(m), $$
where $\hat{df}(\lambda_n)$ is the number of non-zero elements in $\bs{\hat{\beta}}_{a}(\lambda_n)$ and $a_n$ is a positive sequence depending on $n$.  It is worth mentioning that $\lambda_n(m+1)$ is optimum in $[\lambda_n(m+1), \lambda_n(m))$ by producing the minimum sum of squared errors (SSE) and the smallest model size concurrently.

Also note that the number of steps can exceed the full model size --- different steps may have a same model size.  Denote $m_k$ the last step having a model size $k$, and $m_k'$ is another step having the same model size.  The theorem also showed that 
$$ \|\b{y} - \X\bs{\hat{\beta}}_a(\lambda_n(m_k))\|^2 < \|\b{y} - \X\bs{\hat{\beta}}_a(\lambda_n(m_k'))\|^2. $$
Theorefore, $\lambda_n(m_k)$ is the overall optimum in $\{\lambda_n: \hat{df}(\lambda_n) = k, \lambda_n \in [0, \infty]\}$.  So the LARS algorithm enables us to create a one-to-one map between a dimension $k$ and the optimum $\lambda_n(m_k)$, 
$$k  \Longleftrightarrow \lambda_n(m_k). $$  

It is easy to see that $\lambda_n(m_{p_0})$ will satisfy condition (A3).  Hence, we have the following corollary from Theorem \ref{thm1}. 

\begin{corollary}
Suppose conditions (A1)--(A2) and (A5) hold, then
\begin{align*}
\lim_{n \to \infty} P^*(\A_n^* = \A \mid p_0) = 1, \\
\lim_{n \to \infty} P^*(\A_n^* = \M_r \mid r) < 1, \ p_0 < r < p_n,
\end{align*}
where $\M_r$ is any $r$-dimensional model.
\label{cor1}
\end{corollary}

This result can also be established for adaptive Elastic-Net.  Denote $\mathcal{T}_n^* = \{j: \hat{\beta}_{aej}^* \neq 0 \}$ an adaptive Elastic-Net estimate of $\A$ using paired bootstrap data.  

\begin{corollary}
Suppose conditions (A1)--(A2) and (A5) hold, then 
\begin{align*}
\lim_{n \to \infty} P^*(\mathcal{T}_n^* = \A \mid p_0) = 1, \\
\lim_{n \to \infty} P^*(\mathcal{T}_n^* = \M_r \mid r) < 1, \ p_0 < r < p_n,
\end{align*}
where $\M_r$ is any $r$-dimensional model.
\label{cor2} 
\end{corollary}

\begin{proof}
It can be proved by using the techniques for deriving Theorem \ref{thm1}, Corollary \ref{cor1} and Theorem \ref{thm2}.  We bypass here. 
\end{proof}


We now study the estimation properties for using residual bootstrap data.  Denote $\mathcal{T}_n^* = \{j: \tilde{\beta}_{aej}^* \neq 0 \}$ an adaptive Elastic-Net estimator of $\A$ using residual bootstrap data. 

\begin{thm} 
Suppose conditions (A1)--(A2) and (A4)--(A5) hold, then 
\begin{align*}
\lim_{n \to \infty} P^*(\mathcal{T}_n^* = \mathcal{A} \mid \lambda_{n1}^+) = 1. 
\end{align*}
Moreover,  let $\lambda_{n1}'$ be another tuning parameter such that the adaptive Elastic-Net estimator under $\lambda_{n1}'$ is of dimension $r$, $p_0<r<p_n$, then 
\begin{align*}
\lim_{n \to \infty} P^*(\mathcal{T}_n^* = \mathcal{M}_{r} \mid \lambda_{n1}') < 1,
\end{align*}
where $\mathcal{M}_{r}$ is any $r$-dimensional model. 
\label{thm2}
\end{thm}

Proofs of Theorem \ref{thm2} are included in Appendix \ref{AP}.  The LARS-EN algorithm for adaptive Elastic-Net estimations is an extension of the LARS algorithm, which shares the same properties of the LARS for deriving Corollaries \ref{cor1}--\ref{cor2}.   Hence we obtain the following corollary from Theorem 2. 

\begin{corollary}
Suppose conditions (A1)--(A2) and (A5) hold, then 
\begin{align*}
\lim_{n \to \infty} P^*(\mathcal{T}_n^* = \A \mid p_0) = 1, \\
\lim_{n \to \infty} P^*(\mathcal{T}_n^* = \M_r \mid r) < 1, \ p_0 < r < p_n,
\end{align*}
where $\M_r$ is any $r$-dimensional model.
\label{cor3}
\end{corollary}

This result can also be established for adaptive LASSO.  Denote $\A_n^* = \{j: \tilde{\beta}_{aj}^* \neq 0 \}$ an adaptive LASSO estimate of $\A$ using residual bootstrap data. 

\begin{corollary}
Suppose conditions (A1)--(A2) and (A5) hold, then 
\begin{align*}
\lim_{n \to \infty} P^*(\A_n^* = \A \mid p_0) = 1, \\
\lim_{n \to \infty} P^*(\A_n^* = \M_r \mid r) < 1, \ p_0 < r < p_n,
\end{align*}
where $\M_r$ is any $r$-dimensional model.
\label{cor4}
\end{corollary}

\begin{proof}
Note that the adaptive LASSO estimator is a special case of the adaptive Elastic-Net estimator with $\lambda_{n2} = 0$.  Theorem \ref{thm2} holds automatically for $\A_n^*$, from which Corollary \ref{cor4} can be deduced. 
\end{proof}

Variable selection consistency of the MF procedure can then be deduced from Corollaries \ref{cor1}--\ref{cor4}.  

\begin{corollary}
Suppose conditions (A1)--(A2) and (A5) hold.  Then the MF procedure is variable selection consistent, e.g.
$$ \lim_{n \to \infty} P(M_{r^*} = \A) = 1,$$
where $M_{r^*}$ is the model selected from the MF procedure. 
\end{corollary}

\begin{proof}
By definition, $\A_n^*$ is an adaptive LASSO estimate of $\A$ using paired or residual bootstrap data.  It is easy to see that 
$$E^*\left({\MF_j \over B}\right) = P^*(\mathcal{A}_n^* = M_j \mid j), \ \lim_{B \to \infty}{\MF_j \over B} = P^*(\mathcal{A}_n^* = M_j \mid j).$$
Combining with Corollaries \ref{cor1} or \ref{cor4}, 
$$ \lim_{n \to \infty} P(\MF_{p_0} > \MF_r) = 1, \quad p_0 < r < p_n.$$
Thus the MF procedure for adaptive LASSO can consistently identify the true dimension and true model via selecting the maximum of $\MF_j, j \in \{1, \dots, p-1\}$, with the highest dimension (if there is a tie).  Similarly, Corollary \ref{cor2} and \ref{cor3} imply variable selection consistency of the MF procedure for adaptive Elastic-Net. 
\end{proof}


However, the MF procedure has potential issues in application.  In Figure~\ref{fig:intro1} (middle) excluding the full model case, the maximum occurs at dimension 1 instead of 3 although their MFs are both close to 1.  In next section, we propose a WMF procedure to tackle this underfitting issue in application.


\section{The WMF procedure}
\label{sec:WMF}
\subsection{Method and Asymptotic properties}

The underfitting issue in MF procedure can be deduced from Corollaries \ref{cor1}--\ref{cor4}.  Take  $\A_n^*$ for an example.  Although it was shown that $\lim_{n \to \infty} P^*(\A_n^* = \A \mid p_0) = 1$,  the conditional probability at some underfit dimensions can also reach one, e.g. $\lim_{n \to \infty} P^*(\mathcal{A}_n^* = M_{r} \mid r) = 1, 0<r<p_0$.  Note that the tuning parameter leading to an underfit $r$-dimensional estimator, denoted as $\lambda_n'$, fulfills $\lambda_n' > \lambda_n$.  Hence, the convergence rate of $P^*(\mathcal{A}_n^* = M_{r} \mid r)$ at some underfit dimensions can exceed the one at the true dimension.  Therefore, the MF procedure would select an underfit model even with a sufficiently large $n$.  

In order to fix things, we introduce a weight to the MF procedure.  An effective weight should be able to down-weight the underfitting MFs asymptotically, i.e. the weight is able to identify underfit dimensions and its effects does not vanish as $n \to \infty$, without significantly up-weighting the overfitting MFs.

\cite{shao1993linear} showed that the overall unconditional (on $\b y$) expected squared prediction error for the OLS estimator of $\bs{\beta}_0$ under model $\alpha$ is 
\begin{align}
T_{\alpha, n} = \sigma^2 + n^{-1}p_\alpha\sigma^2 + \Delta_{\alpha, n}, 
\label{eqt6}
\end{align}
where $p_\alpha$ indicates the size of $\alpha$,  
$\Delta_{\alpha, n} = \bs{\beta}_0^T \b X^T(\b I - \b P_{\alpha}) \b X \bs{\beta}_0 / n$, \\
$\b P_\alpha = \b X_\alpha(\b X_\alpha^T \b X_\alpha)^{-1} \b X_\alpha^T$, $\b X_\alpha$ is a sub-matrix of $\b X$ whose columns are indexed by the components of $\alpha$ and $\b{I}$ is an identity matrix. 

When $\alpha$ is a true or overfit model, it has $\b X \bs{\beta}_0 = \b X_\alpha \bs \beta_\alpha$ and thus 
\begin{align}
\Delta_{\alpha, n} = 0. 
\label{eqt7}
\end{align}

However, if $\alpha$ is an underfit model,  then $\Delta_{\alpha, n} > 0$ for any fixed $n$.  He further assumed that 
\begin{align}
\lim \inf_{n \to \infty} \Delta_{\alpha, n} > 0,
\label{eqt8}
\end{align}
which is argued in the paper to be a minimal type of asymptotic model identifiability condition.  Under assumption (\ref{eqt8}) and by (\ref{eqt6})--(\ref{eqt7}), 
\begin{align}
\lim_{n \to \infty} {T_{\nu, n} \over T_{\kappa, n}} > 1, 
\label{eqt11}
\end{align}
where $\nu$ is an underfit model and $\kappa$ is a true or overfit model.  By (\ref{eqt11}) a formula inversely proportional to $T_{\alpha, n}$ will be an ideal choice for the weight.

\cite{rao1997out} proposed such a formula for estimating the posterior probability of the model size given the data
\begin{eqnarray}
\hat{P}(j \mid \mathbf{y}) = {\exp[-\hat{T}_n(j) / c\sigma^2] \over \sum_{j=1}^p \exp[-\hat{T}_n(j) / c\sigma^2]}, 
\label{eqt9}
\end{eqnarray}
where $\hat{T}_n(j)$ is an estimate of $T_{\alpha, n}$ using a $j$-dimensional model and $c, 1 \le c \le 2$, is a constant.  We use the multi-fold CV for $\hat{T}_n(j)$ and define 
\begin{eqnarray} 
\WMF_j = \hat{P}(j \mid \mathbf{y}) \times \MF_j. 
\label{eqt10}
\end{eqnarray}
Figure~\ref{fig:intro1} (bottom) shows the effect of weights in Example 1,  which heavily punish underfitting MFs and have little effect on true and overfitting MFs. The WMF procedure then selects the dimension $r^*$ and model $M_{r^*}$ s.t.
$$r^* = \arg \max_{1 \le j \le p-1}  \WMF_j.$$

Recall that $\MF_j/B$ is a bootstrap version estimate of the posterior probability of model $M_j$ given the data and dimension, i.e. $P(M_j \mid \mathbf{y}, j)$,  along with (\ref{eqt10}) it has  
$$\WMF_j = \hat{P}(j \mid \mathbf{y}) \times \hat{P}(M_j \mid \mathbf{y}, j)=\hat{P}(M_j \mid \mathbf{y}).$$
Note that BIC is a Laplace approximation to $P(M_j \mid \mathbf{y})$ under a flat prior assumption and is variable selection consistent for adaptive LASSO \cite{wang2007unified,wang2009shrinkage}, but no convergence rate has been studied.  Simulation studies in Section~\ref{sec:simu} show that BIC has a much slower empirical convergence rate than the WMF procedure.  

Next we show properties of the multi-fold CV using adaptive LASSO or adaptive Elastic-Net estimators.  Then variable selection consistency of the WMF procedure can be established.  Let $K$ be a fixed integer and suppose $n = Kt$.   In multi-fold CV,  one randomly divides a sample of $n$ observations into $K$ mutually exclusive subgroups $s_1, \dots, s_K$ with each subgroup containing $t$ observations, and selects the model by minimizing the following sum of squared errors  
$$ \MCV_{\M} = {1 \over n} \sum_{i=1}^K \|\y_{s_i} - \X_{s_i, \M} \bs{\hat{\beta}}_{s_i^c, \M} \|^2, $$
where $\bs{\hat{\beta}}_{s_i^c, \M}$ is an adaptive LASSO or adaptive Elastic-Net estimator under model $\M$ using samples not in $s_i$.  Let $\alpha$ and $\alpha'$ be the true or overfit models and $\nu$ be an underfit model.  We assume following condition for asymptotic studies of the multi-fold CV procedure. 

\vspace{2.5mm}

(A6) $\sup_{t \to \infty}\sup_{s_i}\|t^{-1}\X_{s_i, \M}^T\X_{s_i, \M} - \b{V}_{\M}\| = o(1),$ where $\b{V}_{\M}$ is a positive definite matrix.

\vspace{3mm}

\begin{thm}
Suppose conditions (A1)--(A2) and (A5)--(A6) hold, then \\
1. the multi-fold CV for adaptive LASSO or adaptive Elastic-Net satisfies
\begin{align*}
&\lim_{n \to \infty} |\MCV_{\a} - \MCV_{\a'}| = \lim_{n \to \infty} \Big|O_p\Big({p_\alpha - p_{\alpha'} \over n}\Big)\Big| = 0, \\
&\lim_{n \to \infty} \MCV_{\nu} - \MCV_{\alpha} 
\ge {d\|\bs{\beta}_{0\nu^c}\|^2 \over 2} + O_p \left(\|\bs{\beta}_{0\nu^c}\| \sqrt{{p_n \over n}}\right) - O_p\left({p_{\a} \over n}\right) > 0,
\end{align*}
2.  model $M_{r^*}$ selected from the $\WMF$ procedure fulfills
$$\lim_{n \to \infty} P(M_{r^*} = \mathcal{A}) = 1.$$ 
\label{thm3}
\end{thm} 

Proofs of Theorem \ref{thm3} are included in Appendix \ref{AP}.  Denote $r'$ an underfit dimension.  The ratio of ${\WMF_{p_0} \over \WMF_{r'}}$ is exponentially proportional to the bias term, ${d \over 2c\sigma^2}\|\bs{\beta}_{0\M_{r'}^c}\|^2$, which is larger than 0 and does not fade as $n \to \infty$.  This guarantees a good finite sample performance of the WMF procedure and a fast vanishing rate of its underfitting issues, which will be confirmed in simulation studies in Section \ref{sec:simu}.

\subsection{Computation}
\label{sec:comp}

In adaptive Elastic-Net, $\lambda_{n2}$ takes the same value in Elastic-Net for calculating the weights $\omega_j$'s, where the tuning parameters are chosen by minimizing the two-dimensional BIC \cite{zou2005regularization}.  Then computational efforts remain the same for adaptive LASSO and adaptive Elastic-Net, which are to compute a full solution path against $\lambda_n$'s or $\lambda_{n1}^+$'s.  Computational complexity of creating an entire adaptive LASSO solution path is of order $O(np_n^2)$ \cite{zou2006adaptive}.  It is of order $O(np_n^2+p_n^3)$ for adaptive Elastic-Net\cite{zou2005regularization}.  Since the optimal value often occurs at an early stage, we could stop the algorithms after $m, m<p_n,$ steps.  In this case, the computational cost reduces to $O(nm^2)$ for adaptive LASSO and $O(m^3+nm^2)$ for adaptive Elastic-net.  

Computational cost of a WMF procedure is then $B$ times the cost of computing an adaptive LASSO or adaptive Elastic-Net solution path.



\begin{figure}
\centering
\includegraphics[width=0.6\textwidth]{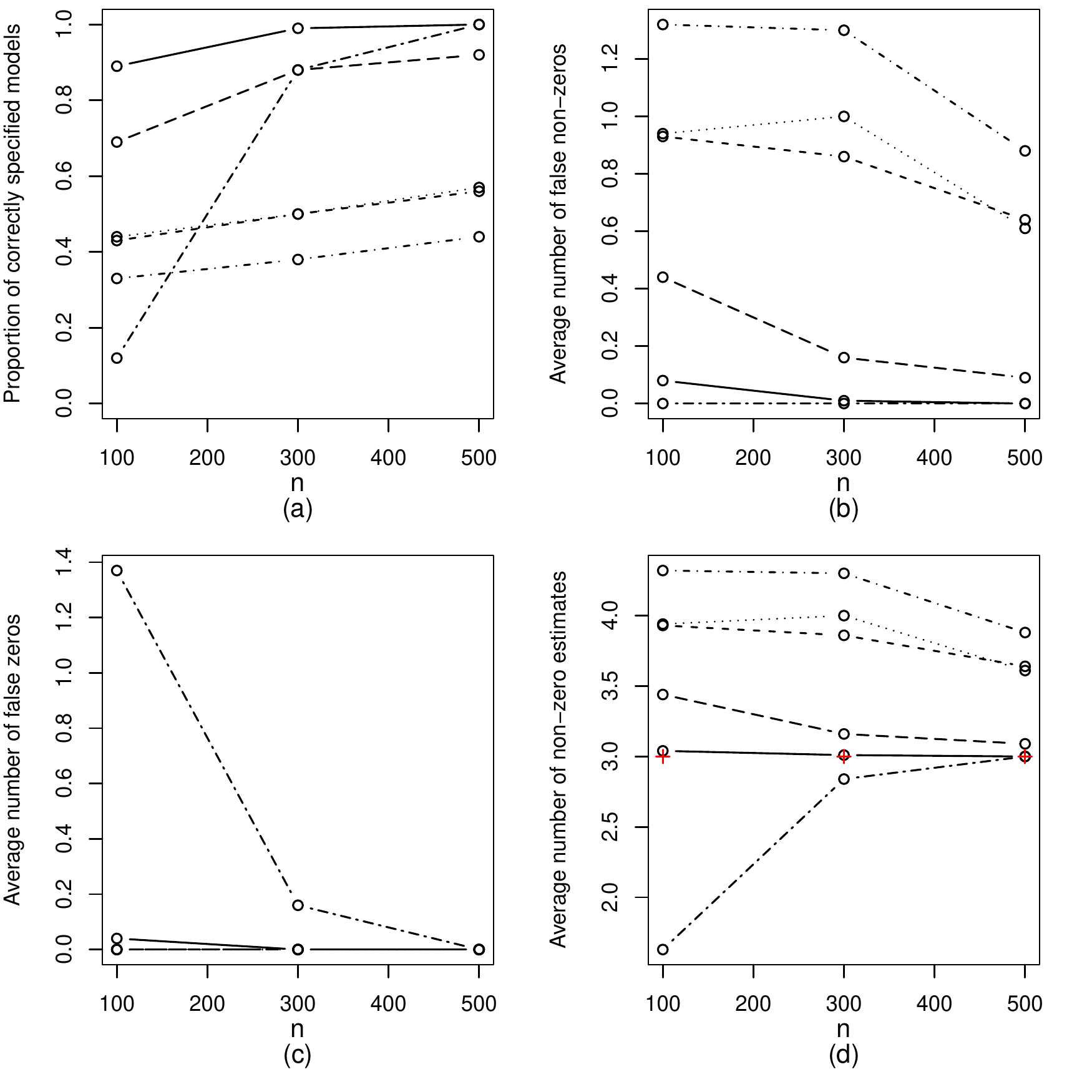}
\includegraphics[width=0.2\textwidth]{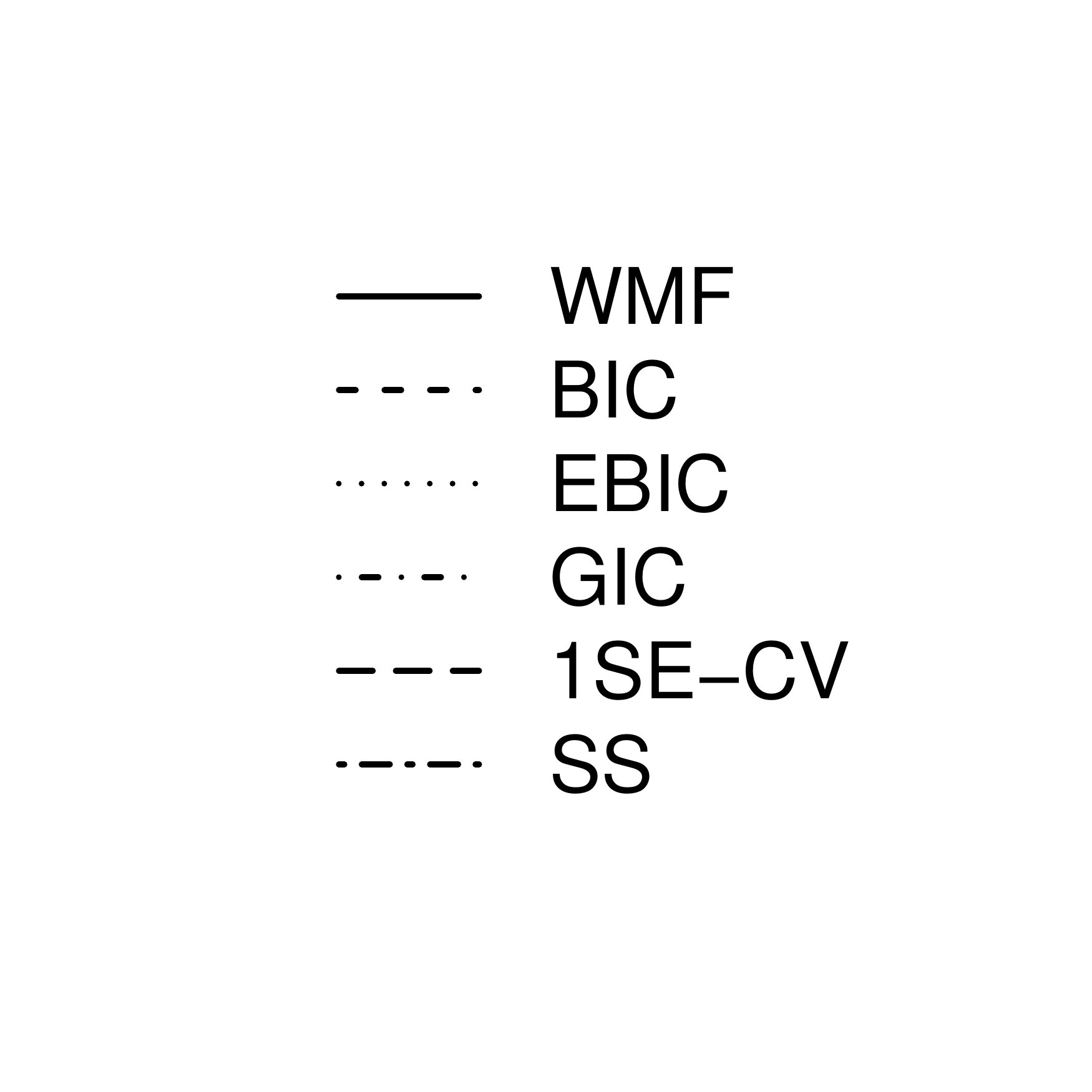}
\caption{Results of scenario 1: (a) proportion of correctly specified models; (b) average number of false non-zeros; (c) average number of false zeros; (d) average value of estimated model sizes.}
\label{fig:case1}
\end{figure}

\section{Empirical studies}
\label{sec:simu}

We now investigate empirical performances of the WMF procedure and show it outperforms the BIC, EBIC, GIC, SS, Cp, and 1se-CV (which is often recommended for variable selection) in a wide range of situations for both adaptive LASSO and adaptive Elastic-Net. The Cp did very poor in all scenarios, thus is excluded in the presentation. 

In all simulations, data were generated from
\begin{eqnarray}
y_i = \b x_i^T{\bs \beta}+\sigma\v_i, \quad i = 1, \dots,  n,
\end{eqnarray}
where $\b x_i \overset{iid}{\sim} N_{p_n}(\b{0}, \bs{\Sigma})$ and $\v_i \overset{iid}{\sim} N(0, 1)$.  Let $p_n=O(n^{\kappa})$ for some constant $\kappa$, $0 \le \kappa <1$, $n = 100, 300, 500$.  Results were averaged over 100 times of replications.  

\begin{figure}
\centering
\includegraphics[width=0.6\textwidth]{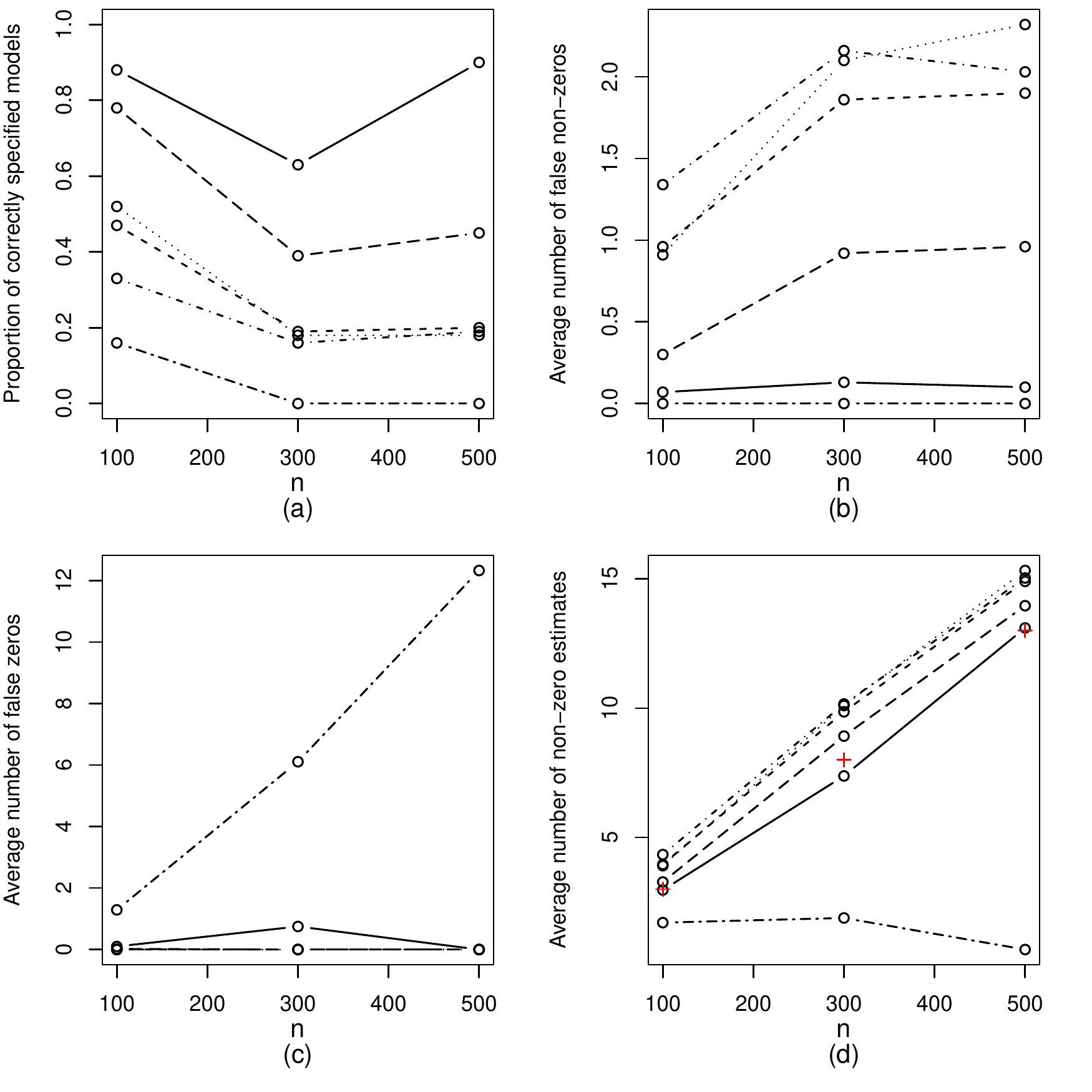}
\includegraphics[width=0.2\textwidth]{legend1.pdf}
\caption{Results of scenario 2: (a) proportion of correctly specified models; (b) average number of false non-zeros; (c) average number of false zeros; (d) average value of estimated model sizes.}
\label{fig:case2}
\end{figure}

\subsection{Simulations of the adaptive LASSO WMF procedure}

Three scenarios were designed for the adaptive LASSO WMF procedure.  In each scenario, $\bs \Sigma(i, j)=0.3^{|i-j|}$ and $\sigma=3$.

Scenario 1: {\it Fixed low dimension and moderate proportion of true covariates.} More specifically, set $p_n=10$ and $\bs{\beta}=(3, 1.5, 0, 0, 2, 0, \dots)_{10}^T$.  Then the proportion of true covariates is 0.3, and the signal to noise ratios (SNR) are respectively 2.03, 2 and 1.98 for various $n$.

Scenario 2: {\it Low dimension, moderate proportion of true covariates and weak signals for some true covariates.} Specifically, set $p_n=O(\sqrt{n})$, then $p_n$ equals to 10, 17, 22 accordingly.  Let $p_0$ grow with $n$ as follows.  Initially $p_0=3$ and $\bs{\beta} = (3, 1.5, 0, 0, 2, 0, \dots)^T$.  Afterwards, $p_0$ increases by 1 for every 40-unit increment in $n$ and the new element equals to 1.  As a result, the proportions of true covariates are respectively 0.3, 0.47, and 0.59, and the SNRs are 2, 2.85 and 3.69.

Scenario 3: {\it High dimension, sparse proportion of true covariates and relatively large signals for all true covariates.} In detail, set $p_n=O(n^{3/4})$, then $p_n$ equals to 32, 72, 106 accordingly.  Let $p_0$ grow in the same manner as in scenario 2, but the new elements equal to 2.  Accordingly, the proportions of true covariates are 0.09, 0.11 and 0.12, and the SNRs are 2, 5.07, and 8.5.

\begin{figure}
\centering
\includegraphics[width=0.6\textwidth]{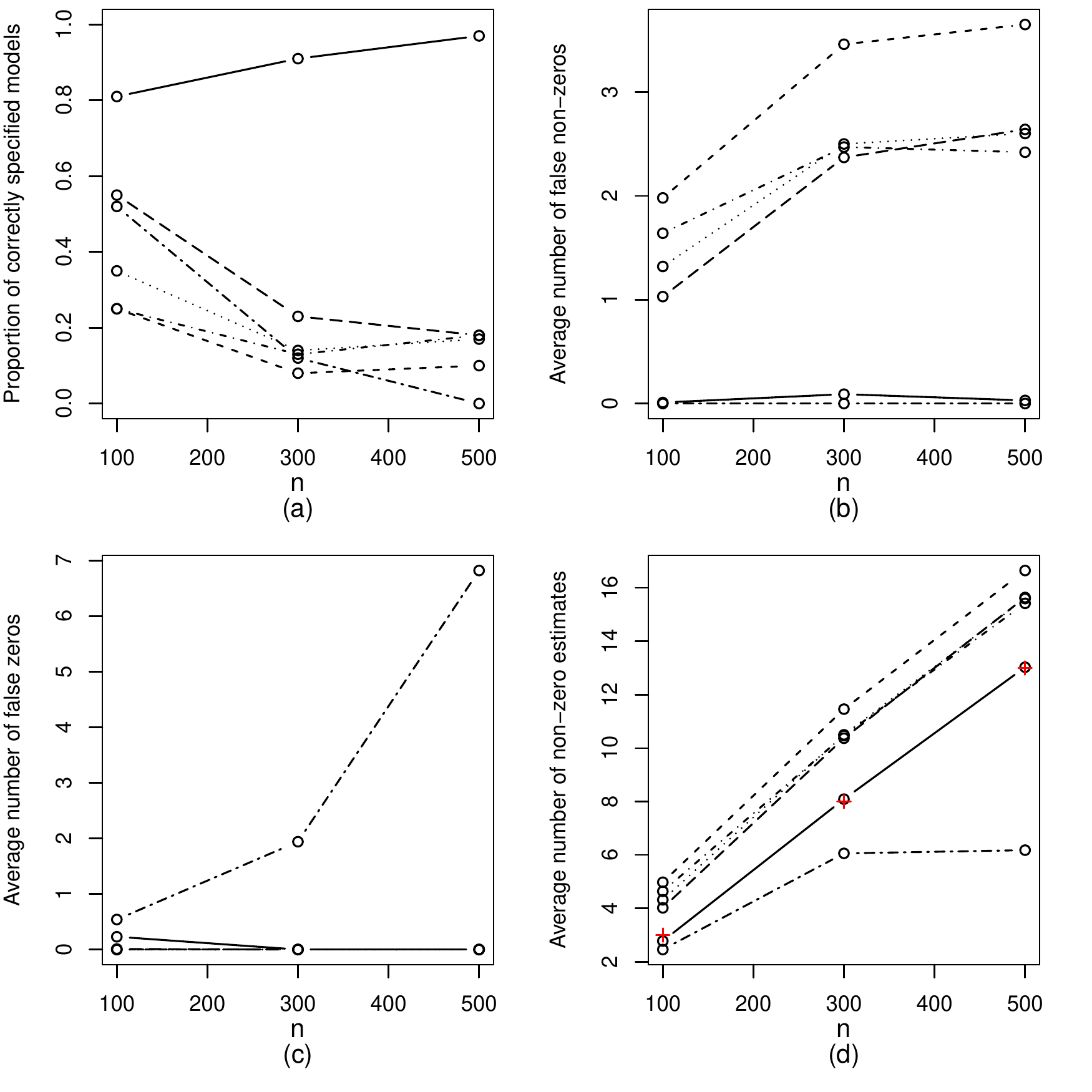}
\includegraphics[width=0.2\textwidth]{legend1.pdf}
\caption{Results of scenario 3: (a) proportion of correctly specified models; (b) average number of false non-zeros; (c) average number of false zeros; (d) average value of estimated model sizes.}
\label{fig:case3}
\end{figure}

{\it Paired bootstrapping} was used in the adaptive LASSO WMF procedure.  Simulation results are summarized in Figures \ref{fig:case1}--\ref{fig:case3}.  In all scenarios, the proposed method has the highest degree of accuracy in identifying the true model and also enjoys a much faster convergence rate than other compared methods.  The WMF procedure has an underfitting issue which vanishes quickly as $n$ increases.  Other methods (except for the SS) however suffer from an overfitting issue. The sparser the model is, the more serious the issue tends to be.  Performance of the SS relies on particular specifications of several unknown parameters.  Although we have followed instructions in \cite{meinshausen2010stability} for setting those parameters throughout the simulations, its performance remains erratic and unsatisfactory.

Simulations for using {\it residual bootstrapping} in the adaptive LASSO WMF procedure were also conducted.  The results are presented in Appendix \ref{sec:asr}, which are similar to those in above paired bootstrapping simulations.

\begin{figure}
\centering
\includegraphics[width=0.6\textwidth]{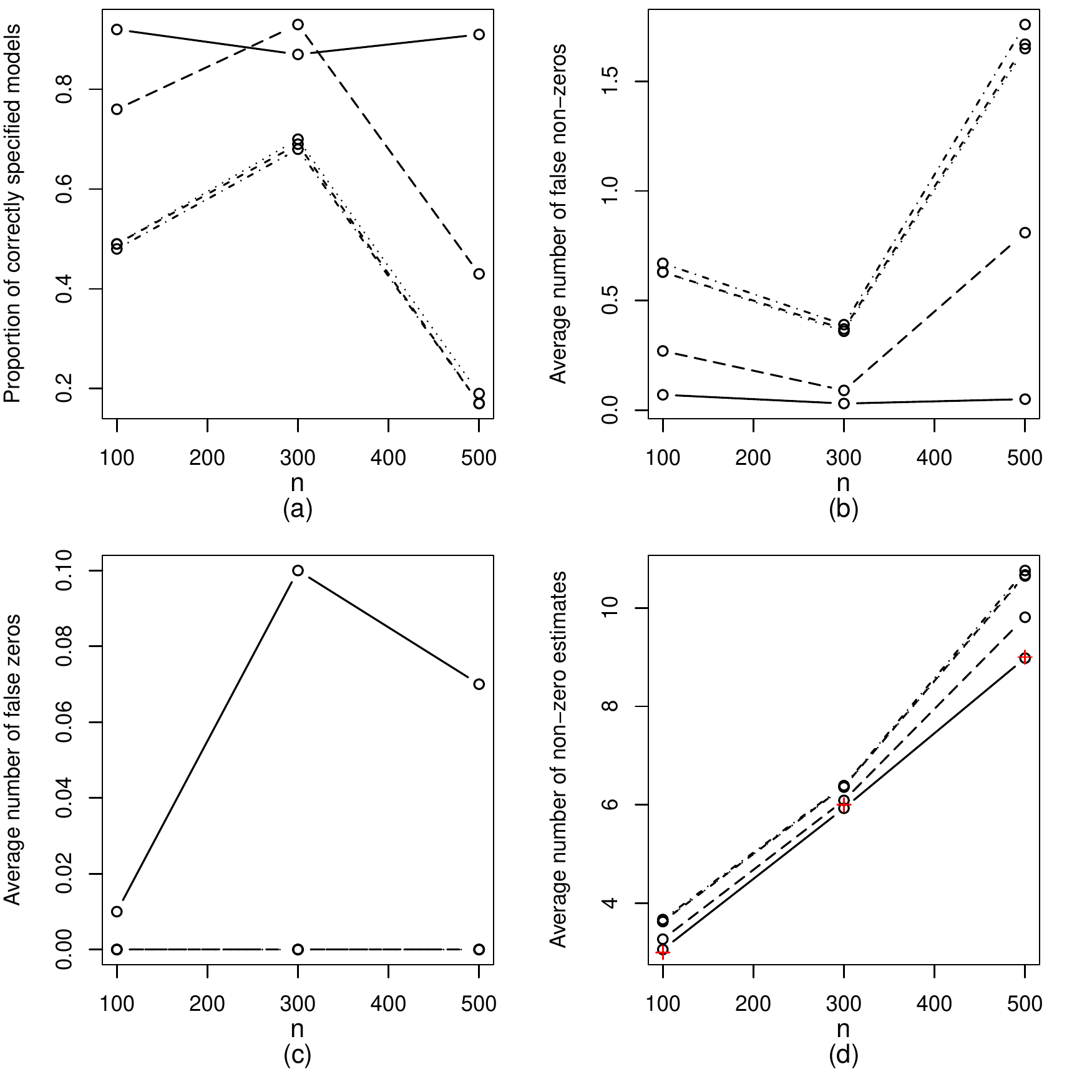}
\includegraphics[width=0.2\textwidth]{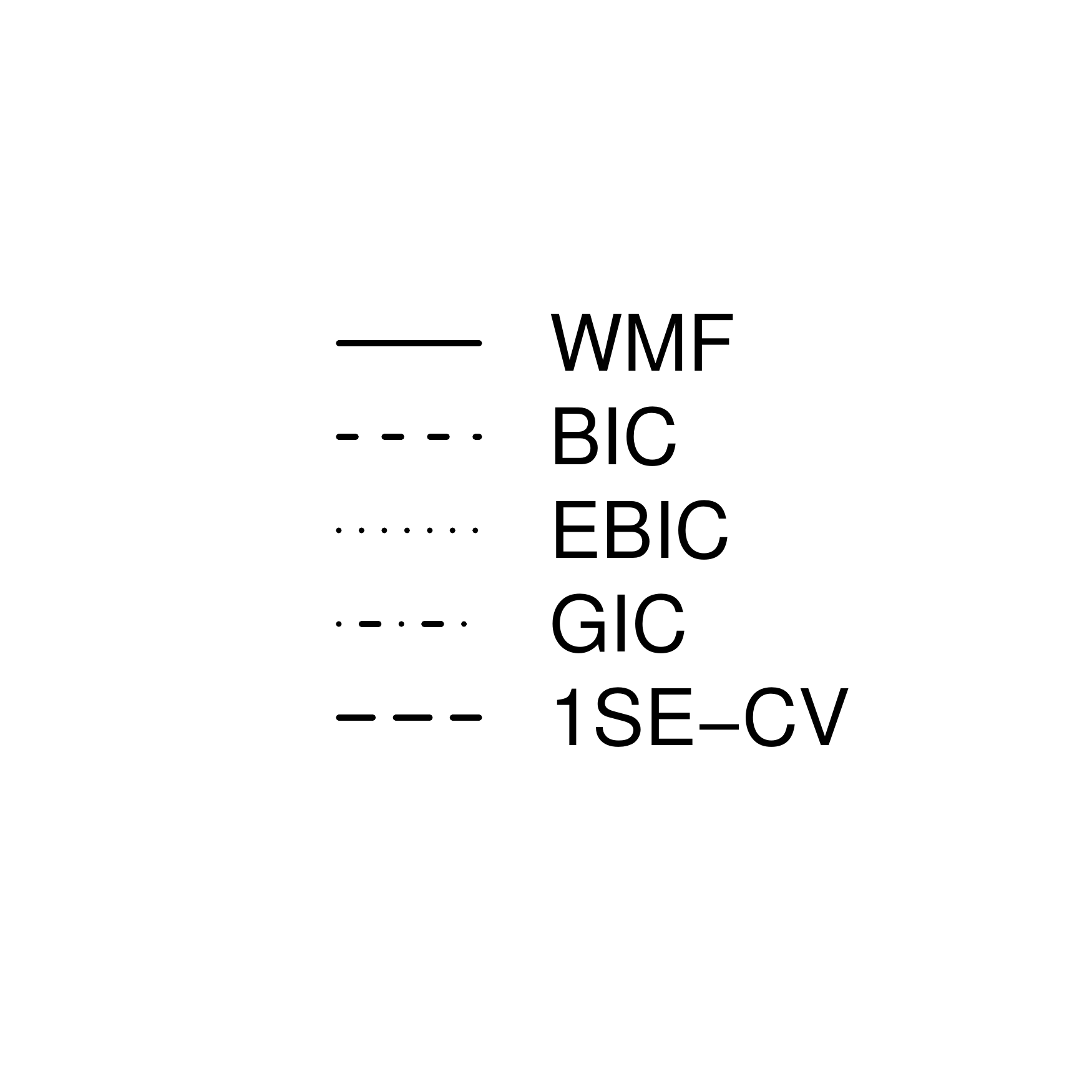}
\caption{Results of scenario 4: (a) proportion of correctly specified models; (b) average number of false non-zeros; (c) average number of false zeros; (d) average value of estimated model sizes.}
\label{fig:case4}
\end{figure}

\subsection{Simulations of the adaptive Elastic-Net WMF procedure}

We also designed three scenarios for the adaptive Elastic-Net WMF procedure, each of which mimics a typical structure in applications.  Since the adaptive Elastic-Net fits data with grouping effects, in following simulations true covariates will be added in blocks with size 3.  The SS is excluded due to its poor performance. 

\begin{figure}
\centering
\includegraphics[width=0.6\textwidth]{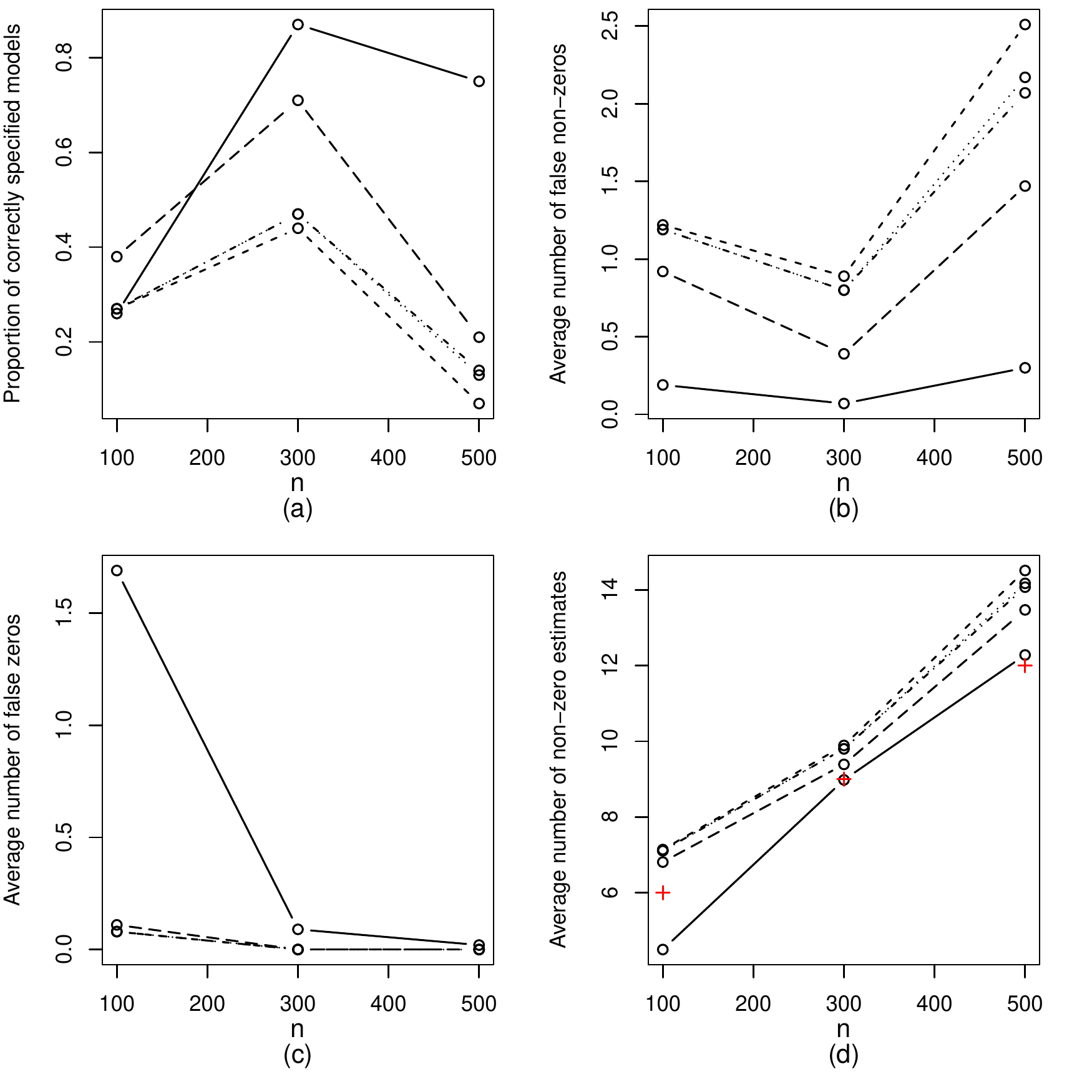}
\includegraphics[width=0.2\textwidth]{legend2.pdf}
\caption{Results of scenario 5: (a) proportion of correctly specified models; (b) average number of false non-zeros; (c) average number of false zeros; (d) average value of estimated model sizes.}
\label{fig:case5}
\end{figure}

Scenario 4: {\it Low dimension, moderate proportion of true covariates, weak signals for some true covariates and moderate correlations between covariates.} More specifically, let $\bs \Sigma (i, j)=0.5^{|i-j|}$, $\sigma=3$, and $p_n=O(\sqrt{n})$.  Initially we have one block of true covariates, then $p_0=3$.  Elements of $\bs{\beta}$ in the block equal to 2, the rest are 0.  Afterwards, we add 1 block of true covariates for every 200-unit increment in $n$ and the new elements equal to 1.  Respectively, the proportions of true covariates are 0.3, 0.35 and 0.41, and the SNRs are 2.45, 3.72 and 3.67.

Scenario 5: {\it High dimension, sparse proportion of true covariates, relatively large signals for all true covariates and moderate correlations between covariates.}  In detail, let $\bs \Sigma(i, j)=0.5^{|i-j|}$, $\sigma=5$, and $p_n=O(n^{3/4})$.  Initially set $p_0=6$.  Then true covariates follow the same adding scheme as in scenario 4.  All non-zero elements in $\bs{\beta}$ equal to 2.  Respectively, the proportions of true covariates are 0.19, 0.13 and 0.11, and the SNRs are 1.79, 3.09 and 3.51.

Scenario 6: {\it High dimension, sparse proportion of true covariates, relatively large signals for all true covariates and high correlations between grouped covariates.} Specifically, let $\sigma=5$ and $p_n=O(n^{3/4})$.  True covariates follow the same adding scheme as in Scenario 5,  all non-zero elements in $\bs{\beta}$ equal to 2.  Moreover, true covariates within each block have correlations almost 1, while true covariates between the blocks have correlation 0.  All noise covariates are i.i.d from $N(0, 1)$.  Respectively, the proportions of true covariates are 0.19, 0.13 and 0.11, and the SNRs are 2.84, 4.35 and 5.75.

\begin{figure}
\centering
\includegraphics[width=0.6\textwidth]{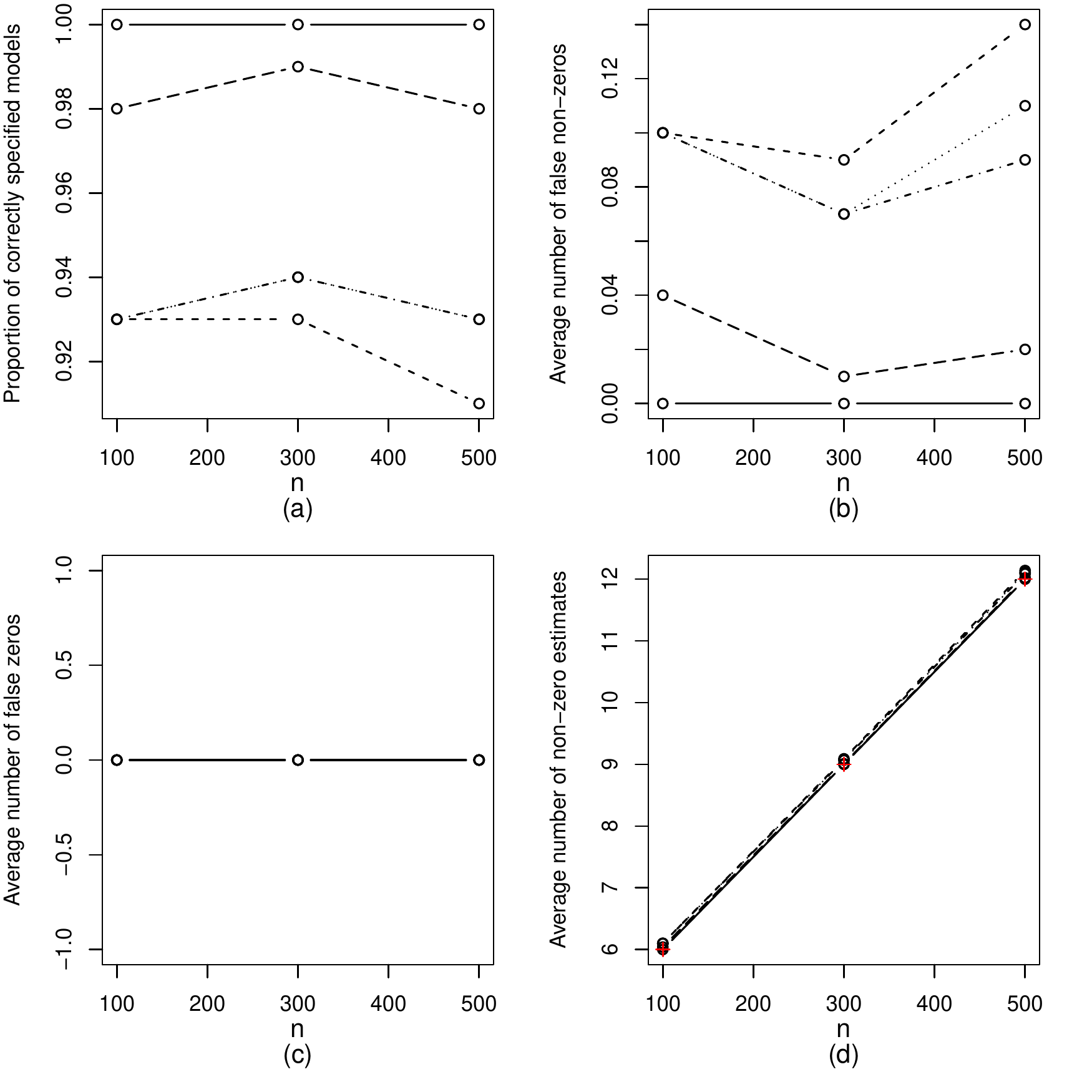}
\includegraphics[width=0.2\textwidth]{legend2.pdf}
\caption{Results of scenario 6: (a) proportion of correctly specified models; (b) average number of false non-zeros; (c) average number of false zeros; (d) average value of estimated model sizes.}
\label{fig:case6}
\end{figure}

{\it Residual bootstrapping} was used in the adaptive Elastic-Net WMF procedure.  Simulation results are summarized in Figures \ref{fig:case4}--\ref{fig:case6}.  In scenarios 4 and 5, the proposed method has the best performance over other compared methods: on average the highest degree of accuracy in indentifying the true model; a faster convergence rate; the underfitting issue vanishes quickly.  On the other hand, other methods suffer from an overfitting issue.  The sparser the model is, the more serious the issue tends to be.   In scenario 6, all methods do equally well because the adaptive Elastic-Net well fit the data with highly grouped effects.    

Simulation results for using {\it paired bootstrapping} in the adaptive Elastic-Net WMF procedure are presented in Appendix \ref{sec:asr}, which are similar to those in above residual bootstrapping simulations.

\begin{table}
\caption{The leukaemia classification using adaptive LASSO \label{tab:1}}
\begin{center}
\begin{tabular}{ |c c c c| } 
\hline
\it Criteria & \it Ten-fold CV error & \it Test error & \it Number of genes \\
\hline
WMF & 0/38 & 5/34 & 5 \\
CV & 0/38 & 4/34 & 13 \\
Cp & 0/38 & 4/34 & 18 \\
BIC & 0/38 & 4/34 & 18 \\
EBIC & 1/38 & 6/34 & 5 \\
GIC & 1/18 & 6/34 & 5 \\
\hline
\end{tabular}
\end{center}
\end{table}

\begin{table}
\caption{The leukaemia classification using adaptive Elastic-Net \label{tab:2}}
\begin{center}
\begin{tabular}{ |c c c c| } 
\hline
\it Criteria & \it Ten-fold CV error & \it Test error & \it Number of genes \\
\hline
WMF & 0/38 & 4/34 & 10 \\
CV & 1/38 & 6/34 & 42 \\
Cp & 1/38 & 6/34 & 36 \\
BIC & 1/38 & 7/34 & 34 \\
EBIC & 1/38 & 7/34 & 34 \\
GIC & 1/38 & 7/34 & 21 \\
\hline
\end{tabular}
\end{center}
\end{table}

\subsection{Classification analysis of the leukaemia data}
We now demonstrate the WMF procedure in a real data application.  
The leukaemia data \cite{golub1999molecular} contains $p_n=7129$ genes and $n=72$ samples. We have 38 out of the 72 samples from the training dataset with 27 ALL's (acute lymphoblastic leukaemia) and 11 AML's (acute myeloid leukaemia).  The remaining 34 samples are from the test dataset with 20 ALL's and 14 AML's.  The goal of this analysis is to identify a subset of genes that can accurately predict the type of leukaemia for future data.  Similar to \cite{zou2005regularization}, we coded the type of leukaemia as a binary response variable, denoted as $y$, and defined the classification function as $I(\hat{y} > 0.5)$, where $I(\cdot)$ is the indicator function.

To improve computational efficiency, we selected 1000 candidate genes as the predictors using the sure independence screening (SIS) procedure \cite{fan2008sure}.  The adaptive LASSO and adaptive Elastic-Net were then applied to explore the data.   The screening and variable selection were carried out on the training dataset,  while classification errors were examined on the test dataset.  Both the LARS and LARS-EN algorithms were stopped after 200 steps of estimation to further reduce the computational costs.  Note that since the optimal steps selected by various types of methods are much smaller than the stopping step, this strategy will not affect the variable selection.

Classification results are summarized in Tables \ref{tab:1}--\ref{tab:2}.  For adaptive LASSO, although the Cp, CV and BIC have obtained the minimal classification errors for both training and test datasets,  the WMF has classification errors close to the minimum using the least number of genes.   For adaptive Elastic-Net, the WMF has the minimal classification errors for both training and test datasets using the least number of genes.  Thus we conclude that the WMF procedure is able to find the set of ``important'' genes that can largely improve the prediction accuracy.

\section{Extensions}
\label{sec:ext}

Here we investigate extensions of the WMF procedure to GLMs, which has the following generic density fuction \cite{mccullagh1989generalized}
$$ f(y \mid \b x, \bs \beta) = h(y)\exp(y \b x^T \bs \beta - \phi(\b x^T \bs \beta)). $$
\cite{zou2006adaptive} had extended the adaptive LASSO to GLMs.  Its estimator, $\bs{\hat{\beta}}_a$, is obtained by maximizing the penalized log-likelihood, 
$$
\bs{\hat{\beta}}_a = \arg \min_{\bs \beta} \sum_{i=1}^n (-y_i \b x_i^T \bs \beta + \phi(\b x_i^T \bs \beta)) + \lambda_n \sum_{j=1}^p \hat{w}_j|\beta_j|, 
$$
where $\hat{w}_j = 1/|\tilde{\beta}_j|^{\gamma}$, $\gamma>0$ and $\bs{\tilde{\beta}} = (\tilde{\beta}_1, \dots, \tilde{\beta}_p)^T$ is the maximum likelihood estimator.  Under certain regularity conditions, $\bs{\hat{\beta}}_a$ was shown to enjoy the oracle properties . 

The generalization of Multi-fold CV to GLMs is straightforward \cite{shao1993linear}.  Define, 
$$
\MCV_{\alpha} = {1 \over n} \sum_{i=1}^k Q(\b y_{s_i},  \b{\hat{y}}_{s_i^c, \alpha}), 
$$
where $Q(\cdot, \cdot)$ is a loss function, $\b{\hat{y}}_{s_i^c, \alpha}$ is the prediction of $\b y_{s_i}$ under model $\alpha$ using samples not in $s_i$.

Then we can extend the WMF procedure to GLMs for adaptive LASSO.  In this case, we draw $B$ paired bootstrap samples in step 1 of Algorithm \ref{alg1}.  Note that the LARS algorithm does not fit for GLMs, but we can use the coordinate descent algorithm \cite{friedman2010regularization} instead, which generates a solution path similar to the LARS.  Hence in step 2, we use the coordinate descent algorithm to fit each bootstrap data.  The rest remain the same.   Asymptotic properties of the adaptive LASSO WMF procedure for GLMs can also be established by using some similar techniques for showing Theorem \ref{thm1} in this paper and Theorem 4 in \cite{zou2006adaptive}.

\begin{figure}
\centering
\includegraphics[width=0.6\textwidth]{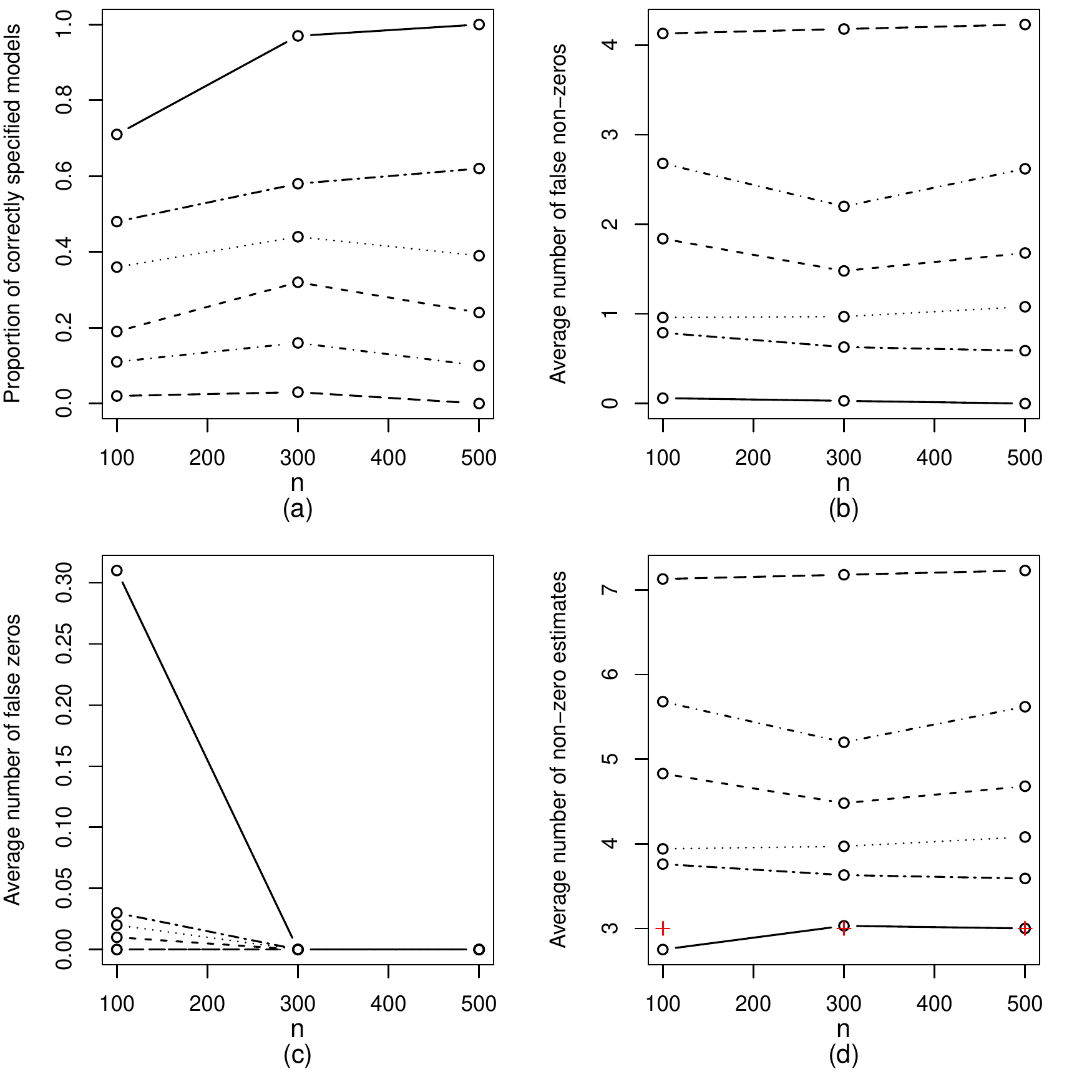}
\includegraphics[width=0.2\textwidth]{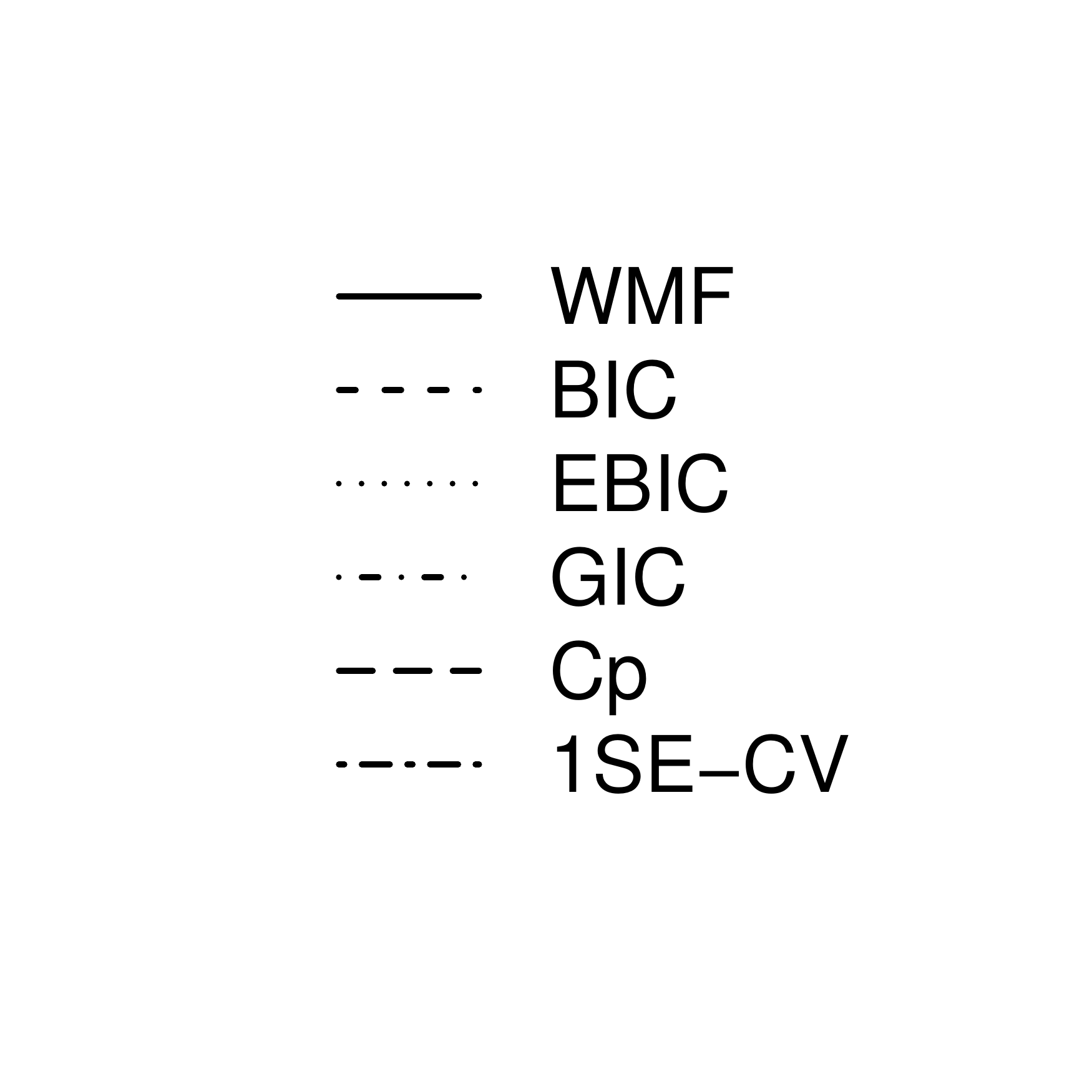}
\caption{Results of the GLM example: (a) correctly specified models; (b) average number of false non-zeros; (c) average number of false zeros; (d) average value of estimated model sizes.}
\label{fig:glm}
\end{figure}

We demonstrate this extension through one simple example, where binary responses were generated from the logistic regression model 
$$ P(y_i \mid \b x_i) = {1 \over 1+exp(- \b x_i^T \bs \beta)}, \quad i=1, \dots, n,$$
where $\b x_i \overset{iid}{\sim} N_{10}(\b{0}, \bs \Sigma)$, $\bs \Sigma(i, j) = 0.3^{|i-j|}$, and $\bs \beta=(3, 1.5, 0,0,2,0,\dots)_{10}^T$.  Simulation results were averaged over 100 times of replications and summarized in Figure~\ref{fig:glm}.  It shows that the WMF procedure is much more accurate in variable selection and also enjoys a faster convergence rate than other compared methods.

Extension of the adaptive Elastic-Net WMF procedure to GLMs is similar.  Define the adaptive Elastic-Net estimator for GLMs as
\begin{align}
\bs{\hat{\beta}}_{ae} = (1+{\lambda_{n2} \over n}) \bigg\{\arg \min_{\bs \beta} &\sum_{i=1}^n (-y_i \b x_i^T \bs \beta + \phi(\b x_i^T \bs \beta))  \nonumber \\
&+ \lambda_{n2}\sum_{j=1}^{p_n}|\beta_j|^2 + \lambda_{n1}^+ \sum_{j=1}^{p_n} w_j|\beta_j|\bigg\}, 
\label{eqt13}
\end{align}
where $w_j = |\hat{\beta}_{ej}|^{-\gamma}$, $\gamma>0$ and $\bs{\hat{\beta}}_e  = (\hat{\beta}_{e1}, \dots, \hat{\beta}_{ep_n})^T$ is defined in (\ref{eqt13}) with $\hat{w}_j=1$ for all $j$'s. The rest follow the same procedures for extension of the adaptive LASSO WMF procedure.

\section{Ultra-high dimensional data}
\label{sec:uhd}

In this section, we discuss applications of the WMF procedure to ultra-high dimensional data in which $p_n > n$.   \cite{fan2008sure} proposed the sure independence screening (SIS) method for ultra-high dimensional data to reduce their dimensionality to a moderate scale, $d_n$, s.t. $d_n < n$.   Afterwards a lower dimensional estimation method such as the SCAD can be applied to the reduced data.  This process is called SIS+SCAD.  Under some regularity conditions, they showed that the SIS has an exponentially small probability to omit true features and the SIS+SCAD retains the oracle properties if $d_n = o_p(n^{1/3})$.  By replacing the SCAD with adaptive Elastic-Net, the new procedure is refered to as SIS+AEnet \cite{zou2009adaptive}, which holds the oracle properties if $d_n = O_p(n^\varrho), 0 \le \varrho < 1$.  Here we recommend to combine SIS with the WMF procedure when $p_n > n$.  We first use the SIS to reduce the dimensionality to $d_n, d_n < n$, and then apply the WMF procedure to the reduced data.  We call this procedure SIS+WMF.

\begin{corollary}
Suppose conditions for Theorem 1 in \cite{fan2008sure} and Theorem 3 in this paper hold.  Let $d_n = n^\varrho, 0 \le \varrho<1$.  Then the SIS+WMF procedure is variable selection consistent. 	
\label{cor5}
\end{corollary}
Note that Corollary \ref{cor5} is a direct conclusion of Theorem 1 in \cite{fan2008sure} and Theorem 3 in this paper.

\section{Discussion}
We proposed a prediction-weighted maximal frequency procedure to estimate the amount of regularization for adaptive LASSO and adaptive Elastic-Net. Asymptotic properties were studied with a diverging $p_n$. 

Central idea of the WMF procedure is the importance of conditioning on dimension, which mitigates overfitting. Underfitting can then be handled by using prediction-based weights estimated by multi-fold cross-validation. This simple recipe can also be applied to other regularization methods, say the SCAD and fused LASSO, making the WMF procedure a unified model selection criterion in regularization problems. However, asymptotic properties have yet to be studied, which will be a future topic.

\appendix
\numberwithin{remark}{section}

\section{Proofs}
\label{AP}

\begin{proof}[Proof of Lemma 1] 
Assume $\mid\beta_i\mid > \mid\beta_j\mid$ and $\mid\beta_i\mid-\mid\beta_j\mid=m\sigma,\ m>0$. We have 4 cases for $\beta_i, \beta_j$
\begin{eqnarray}
\beta_i=
\begin{cases}
\beta_j+m\sigma\quad \mathrm{or}\quad -\beta_j-m\sigma, &\quad \beta_j \ge 0, \\
-\beta_j+m\sigma\quad \mathrm{or}\quad \beta_j-m\sigma, &\quad \beta_j < 0.  \nonumber
\end{cases}
\end{eqnarray}

 Let $Z_i=\beta_i+{x}_i^T\varepsilon \sim N(\beta_i, \sigma^2)$ and $Z_j=\beta_j+{x}_j^T\varepsilon \sim N(\beta_j, \sigma^2)$. We have
\begin{align}
&P(\mid Z_i\mid \le z)=\Phi(\frac{z-\beta_i}{\sigma})+\Phi(\frac{z+\beta_i}{\sigma})-1,\quad z \ge 0, \nonumber \\ 
&P(\mid Z_j\mid \le z)=\Phi(\frac{z-\beta_j}{\sigma})+\Phi(\frac{z+\beta_j}{\sigma})-1,\quad z \ge 0. \nonumber
\end{align}

Consider case 1:\quad $\beta_j \ge 0$ and $\beta_i=\beta_j+m\sigma$, $m > 0$.

Let $k$ be a positive constant. The point $\beta_j+k\sigma$ separates the domain of $Z_i$ and $Z_j$ into two parts: $(-\infty,\; \beta_j+k\sigma]$ and $(\beta_j+k\sigma,\; \infty]$. The cumulative probabilities of $Z_i$ and $Z_j$ in first part of the domain are respectively
\begin{align}
&P(\mid Z_i\mid \le \beta_j+k\sigma)=\Phi(k-m)+\Phi(m+k+\frac{2\beta_j}{\sigma})-1, \nonumber \\
&P(\mid Z_j\mid \le \beta_j+k\sigma)=\Phi(k)+\Phi(k+\frac{2\beta_j}{\sigma})-1. \nonumber
\end{align}

The probability $P(\mid Z_i\mid>\mid Z_j\mid)$ can then be calculated from
\begin{align}
P(\mid Z_i\mid>\mid Z_j\mid) & = {1/2}P(\mid Z_i\mid \le \beta_j+k\sigma, \mid Z_j\mid \le \beta_j+k\sigma) \nonumber \\
& + {1/2}P(\mid Z_i\mid > \beta_j+k\sigma, \mid Z_j\mid > \beta_j+k\sigma) \nonumber \\
& + P(\mid Z_i\mid > \beta_j+k\sigma, \mid Z_j\mid \le \beta_j+k\sigma). \nonumber
\end{align}

After some simple deductions, we get,
\begin{align}
&P(\mid Z_i\mid>\mid Z_j\mid) \nonumber \\
= &\frac{1}{2} \left \{\Phi(k+\frac{2\beta_j}{\sigma})+\Phi(k)-\Phi(k-m)-\Phi(k+m+\frac{2\beta_j}{\sigma}) \right \} +\frac{1}{2}. 
\label{eq20}
\end{align}

If $m=0$ i.e. $\mid\beta_i\mid=\mid\beta_j\mid$, from (\ref{eq20}) we have 
$$P(\mid Z_i\mid>\mid Z_j\mid)=\frac{1}{2}.$$

However if $m>0$ i.e. $\mid\beta_i\mid>\mid\beta_j\mid$, 
\begin{align}
&P(\mid Z_i\mid>\mid Z_j\mid) \nonumber \\
= &\frac{1}{2} \left \{ \int_{k-m}^k \frac{1}{\surd{(2\pi)}}e^{-x^2/2}\, dx - \int_{k+2\beta_j/\sigma}^{k+2\beta_j/\sigma+m} \frac{1}{\surd{(2\pi)}}e^{-x^2/2}\, dx \right \} + \frac{1}{2}.
\label{eq21}
\end{align}

Since $m, k, \beta_j, \sigma>0$, we have 
$$\max \left \{ \mid k-m\mid,\; \mid k\mid \right \} < \max \left \{ \mid k+(2\beta_j)/{\sigma}\mid,\; \mid k+(2\beta_j)/{\sigma}+m\mid \right \}.$$
Note that two integrals in (\ref{eq21}) have equal length of the integral intervals. Moreover the integral function is an monotonically decreasing function of $x$ for $x \ge 0$, and monotonically increasing for $x < 0$. Hence 
\begin{eqnarray}
\int_{k-m}^k \frac{1}{\surd{(2\pi)}}e^{-x^2/2}\, dx - \int_{k+2\beta_j/\sigma}^{k+2\beta_j/\sigma+m} \frac{1}{\surd{(2\pi)}}e^{-x^2/2}\, dx >0.
\label{eq22}
\end{eqnarray}
Combining (\ref{eq21}) with (\ref{eq22}), we get 
$$ P(\mid Z_i\mid>\mid Z_j\mid) > \frac{1}{2}.$$  
Other three cases can be proved in the same way.  We avoid the repetitions here. 
\end{proof}

\begin{proof}[Proof of Theorem 1]

By \cite{zou2009adaptive},  $\bs{\hat{\beta}}_a$ enjoys the oracle properties under certain regularity conditions.  And $\bs{\hat{\beta}}_a^*$ is a paired bootstrap analog of $\bs{\hat{\beta}}_a$ by replacing $(\X, \y)$ with $(\X^*, \y^*)$ in estimation.  To simplify notations in the proof, we drop the subscript `$a$' in $\bs{\hat{\beta}}_a$ and $\bs{\hat{\beta}}_a^*$. 

By the KKT regularity conditions, $\bs{\hat{\beta}}^*$ is the unique solution of adaptive LASSO given $(\b{X}^*, \b{y}^*)$ if
\begin{align}
\begin{cases}
\X_j^{*T} (\y^* - \X^* \bs{\hat{\beta}}^*) = \lambda_n \omega_j sgn(\hat{\beta}_j^*), \quad \hat{\beta}_j^* \neq 0 \\
|\X_j^{*T}(\y^* - \X^* \bs{\hat{\beta}}^*)| < \lambda_n \omega_j, \quad \hat{\beta}_j^* = 0
\end{cases}
\label{eq1}
\end{align}
where $\X_j^*$ is the $j$th column of $\b{X}^*$ and 

\begin{align*}
sgn(x) = 
\begin{cases}
1, \quad x > 0, \\
0, \quad x = 0, \\
-1, \quad x < 0.
\end{cases}
\end{align*}

Let $\b{\tilde{s}}_{\A} = (\omega_j sgn(\hat{\beta}_j), j \in \A)^T$ and $\bs{\hat{\beta}}_{\A}^* = (\X_{\A}^{*T} \X_{\A}^*)^{-1} (\X_{\A}^{*T} \y^* - \lambda_n \b{\tilde{s}}_{\A})$.  We show that $(\bs{\hat{\beta}}_{\A}^*, \b{0})$ satisfies (\ref{eq1}) with probability tending to 1, which is equivalent to prove
\begin{align}
\begin{cases}
sgn(\hat{\beta}_j) (\hat{\beta}_j - \hat{\beta}_j^*) < |\hat{\beta}_j|, \quad j \in \A, \\
|\X_j^{*T}(\y^* - \X_{\A}^* \bs{\hat{\beta}}_{\A}^*)| < \lambda_n \omega_j, \quad j \notin \A,
\end{cases}
\label{eq2}
\end{align} 
where the first inequation implies $sgn(\bs{\hat{\beta}}_{\A}^*) =  sgn(\bs{\hat{\beta}}_{\A})$. 

Note that $\omega_j = |\tilde{\beta}_j|^{-\gamma}$, where $\bs{\tilde{\beta}} = (\tilde{\beta}_1, \dots, \tilde{\beta}_{p_n})^T$ is an OLS or best ridge estimate of $\bs{\beta}_0$,  
$$ \bs{\tilde{\beta}}(\lambda_{n2}) = \arg\min_{\bs{\beta}} \|\y - \X\bs{\beta}\|^2 + \lambda_{n2}\sum_{j=1}^{p_n}|\beta_j|^2. $$
By Theorem 3.1 in \cite{zou2009adaptive}, 
\begin{align}
E\|\bs{\tilde{\beta}}(\lambda_{n2}) - \bs{\beta}_0\|^2 \le 2{\lambda_{n2}^2\|\bs{\beta}_0\|^2 + np_nD\sigma^2 \over (nd+\lambda_{n2})^2} = O_p\left({p_n \over n}\right)
\label{eq18}
\end{align}
under assumption that $\lim_{n \to \infty}{\lambda_{n2} \over \sqrt{n}} = 0$.  It is satisfied automatically for the OLS estimate.

Denote $\x_{i\A}^*$ the $i$th row of $\X_{\A}^*$, and $\otimes$ the element-wise product.  We have 
\begin{align*}
\bs{\hat{\beta}}_{\A}^* - \bs{\hat{\beta}}_{\A} &= (\X_{\A}^{*T} \X_{\A}^*)^{-1} (\X_{\A}^{*T} \y^* - \X_{\A}^{*T} \X_{\A}^* \bs{\hat{\beta}}_{\A} - \lambda_n \b{\tilde{s}}_{\A}) \\
&= (\X_{\A}^{T} \X_{\A})^{-1} \left [\sum_{i=1}^n \x_{i\A}^* (y_i^* - \x_{i\A}^{*T} \bs{\hat{\beta}}_{\A}) - \lambda_n \bs{\omega}_{\A} \otimes sgn(\bs{\hat{\beta}}_{\A}) \right] (1+o_p(1)).
\end{align*}

Hence under conditions (A1) and (A5), 
\begin{align*}
E^*\|\bs{\hat{\beta}}_{\A}^* - \bs{\hat{\beta}}_{\A}\|^2 &\le {E^* \left \|\sum_{i=1}^n \left[\x_{i\A}^* (y_i^* - \x_{i\A}^{*T} \bs{\hat{\beta}}_{\A}) - {\lambda_n \bs{\omega}_{\A} \over n} \otimes sgn(\bs{\hat{\beta}}_{\A}) \right] \right \|^2 \over \zeta_{min}^2 (\X_{\A}^{T} \X_{\A})} \\
&= {\sum_{i=1}^n E^* \left\| \x_{i\A}^* (y_i^* - \x_{i\A}^{*T} \bs{\hat{\beta}}_{\A}) - {\lambda_n \bs{\omega}_{\A} \over n} \otimes sgn(\bs{\hat{\beta}}_{\A}) \right \|^2 \over \zeta_{min}^2 (\X_{\A}^{T} \X_{\A})} \\
&= {\sum_{i=1}^n \left\| \x_{i\A} (y_i - \x_{i\A}^{T} \bs{\hat{\beta}}_{\A}) - {\lambda_n \bs{\omega}_{\A} \over n} \otimes sgn(\bs{\hat{\beta}}_{\A}) \right \|^2 \over \zeta_{min}^2 (\X_{\A}^{T} \X_{\A})} \\
&\le {1 \over (nd)^2} \left[ \sum_{i=1}^n 2 \x_{i\A}^T \x_{i\A} (y_i - \x_{i\A}^{T} \bs{\hat{\beta}}_{\A})^2 + {2 \lambda_n^2 \|\bs{\omega}_{\A}\|^2 \over n} \right] \\
&\le {2p_0D\sigma^2 \over nd^2} + {2\lambda_n^2\|\bs{\omega}_{\A}\|^2 \over n^3d^2}. 
\end{align*}

Let $\psi = \min_{j \in \A}|\beta_{0j}|$, $\tilde{\psi} = \min_{j \in \A}|\tilde{\beta}_j|$ and $\hat{\psi} = \min_{j \in \A}|\hat{\beta}_j|$.  Under conditions (A1)--(A3) and (A5), the first inequation in (\ref{eq2}) can be proved by 
\begin{align*}
&P^* \left\{ \exists j \in \A, sgn(\hat{\beta}_j) (\hat{\beta}_j - \hat{\beta}_j^*) \ge |\hat{\beta}_j| \right\} \\
\le &\sum_{j \in \A} P^* \left\{ sgn(\hat{\beta}_j) (\hat{\beta}_j - \hat{\beta}_j^*) \ge |\hat{\beta}_j|, \tilde{\psi}>\psi/2, \hat{\psi}>\psi/2 \right\} \\
&+ P(\tilde{\psi} \le \psi/2) + P(\hat{\psi} \le \psi/2) + P(\tilde{\psi} \le \psi/2, \hat{\psi} \le \psi/2) \\
\le & {4E^*\big(\|\bs{\hat{\beta}}_{\A}^* - \bs{\hat{\beta}}_{\A}\|^2 I(\tilde{\psi}>\psi/2)\big) \over \psi^2} + c_1 + c_2 + \min\{c_1, c_2\} \\
\le & {8 \over \psi^2} \left( {p_0D\sigma^2 \over nd^2} + {\lambda_n^2p_0(\psi/2)^{-2\gamma} \over n^3d^2} \right) + c_1 + c_2 + \min\{c_1, c_2\} \\
= &O_p\left({p_0 \over n\psi^2}\right) + o_p\bigg(\Big({\lambda_n \over \sqrt{n}\psi^\gamma}\Big)^2{p_0 \over n\psi^2}\bigg) + c_1 + c_2 +\min\{c_1, c_2\} \\
\to &0,    
\end{align*}
where
$$ c_1 \le P(\|\bs{\tilde{\beta}} - \bs{\beta}_0\| \ge \psi/2) \le {4E\|\bs{\tilde{\beta}} - \bs{\beta}_0\|^2 \over \psi^2}. $$
By (\ref{eq18}), it has
$$ c_1 \le 8{\lambda_{n2}^2\|\bs{\beta}_0\|^2 + np_nD\sigma^2 \over \psi^2(nd+\lambda_{n2})^2} = O_p\left({p_n \over n\psi^2}\right) \to 0$$
Similarly, 
\begin{align*}
c_2 &\le P(\|\bs{\hat{\beta}}_{\A} - \bs{\beta}_{0\A}\| \ge \psi/2) \le {4E\|\bs{\hat{\beta}}_{\A} - \bs{\beta}_{0\A}\|^2I(\tilde{\psi}>\psi/2) \over \psi^2} + c_1.
\end{align*}
By Theorem 3.1 in \cite{zou2009adaptive},  
\begin{align}
c_2 &\le 16{np_nD\sigma^2 + \lambda_n^2p_0(\psi/2)^{-2\gamma} \over \psi^2n^2d^2} + c_1 \nonumber \\
&= O_p\left({p_n \over n\psi^2}\right) + O_p\bigg(\Big({\lambda_n \over \sqrt{n}\psi^\gamma}\Big)^2{p_0 \over n\psi^2}\bigg) \nonumber \\
&\to 0.
\label{eq23}
\end{align}

For proof of the second inequation in (\ref{eq2}), it suffices to show
$$ P^* \left\{ \exists j \notin \A, |\X_j^{*T}(\y^* - \X_{\A}^* \bs{\hat{\beta}}_{\A}^*)| \ge \lambda_n \omega_j \right\} \to 0. $$

Since 
$$ |\X_j^{*T}(\y^* - \X_{\A}^* \bs{\hat{\beta}}_{\A}^*)| \le |\X_j^{*T}(\y^* - \X_{\A}^* \bs{\hat{\beta}}_{\A})| + |\X_j^{*T}\X_{\A}^*(\bs{\hat{\beta}}_{\A} - \bs{\hat{\beta}}_{\A}^*)|, $$
it follows that 
\begin{align*}
&P^* \left\{ \exists j \notin \A, |\X_j^{*T}(\y^* - \X_{\A}^* \bs{\hat{\beta}}_{\A}^*)| \ge \lambda_n \omega_j \right\} \\
\le &\sum_{j \notin \A} P^* \left\{ |\X_j^{*T}(\y^* - \X_{\A}^* \bs{\hat{\beta}}_{\A})| \ge (1-\kappa) \lambda_n \omega_j \right\}  \\
&+ \sum_{j \notin \A} P^* \left\{ |\X_j^{*T}\X_{\A}^*(\bs{\hat{\beta}}_{\A} - \bs{\hat{\beta}}_{\A}^*)| \ge \kappa \lambda_n \omega_j \right\} \\
= &B_1 + B_2,
\end{align*}
where $\kappa$, $0<\kappa<1$, is a constant.

For $B_1$, 
\begin{align*}
&\sum_{j \notin \A} E^* |\X_j^{*T}(\y^* - \X_{\A}^* \bs{\hat{\beta}}_{\A})|^2 = \sum_{j \notin \A} E^* \left|\sum_{i=1}^n x_{ij}^*(y_i^* - \x_{i\A}^{*T}\bs{\hat{\beta}}_{\A}) \right|^2 \\
= &\sum_{j \notin \A} E^* \left[ \sum_{i=1}^n x_{ij}^{*2}(y_i^* - \x_{i\A}^{*T}\bs{\hat{\beta}}_{\A})^2 + \sum_{i \neq k} x_{ij}^*(y_i^* - \x_{i\A}^{*T}\bs{\hat{\beta}}_{\A}) x_{kj}^*(y_k^* - \x_{k\A}^{*T}\bs{\hat{\beta}}_{\A}) \right] \\
= &\sum_{j \notin \A} \left\{ \sum_{i=1}^n x_{ij}^{2}(y_i - \x_{i\A}^{T}\bs{\hat{\beta}}_{\A})^2 + n(n-1)\left[ {1 \over n} \sum_{i=1}^n x_{ij}(y_i - \x_{i\A}^{T}\bs{\hat{\beta}}_{\A}) \right]^2 \right\} \\
= &np_{\A^c}\sigma^2 + {n-1 \over n} \|\X_{\A^c}^{T}(\y - \X_{\A} \bs{\hat{\beta}}_{\A})\|^2 \\
\le &np_{\A^c}\sigma^2 + (n-1)p_{\A^c}D\sigma^2, 
\end{align*}
where $p_{\A^c}$ indicates the size of $\A^c$.  By (\ref{eq18}), $\forall j \in \A^c$, $E|\tilde{\beta}_j|^2 \le E\|\bs{\tilde{\beta}} - \bs{\beta}_0\|^2 = O_p\left({p_n \over n}\right)$, which indicates $|\tilde{\beta}_j| \le O_p\left({p_n \over n}\right)^{1/2}$.  Then under condition (A3), $B_1$ fulfills
\begin{align*}
B_1 &\le \sum_{j \notin \A} {E^* |\X_j^{*T}(\y^* - \X_{\A}^* \bs{\hat{\beta}}_{\A})|^2  \over (1-\kappa)^2\lambda_n^2\omega_j^2} \\
&\le {np_{\A^c}\sigma^2 + (n-1)p_{\A^c}D\sigma^2 \over (1-\kappa)^2\lambda_n^2O_p\left({p_n \over n}\right)^{-\gamma}} \\
&= O_p\left({n \over \lambda_n^2 n^{(1-\varrho)(1+\gamma)-1}}\right) \\
&\to 0.
\end{align*}

Also since 
\begin{align*} 
&\sum_{j \notin \A} E^* \big(|\X_j^{*T}\X_{\A}^*(\bs{\hat{\beta}}_{\A} - \bs{\hat{\beta}}_{\A}^*)|^2I(\tilde{\psi}>\psi/2)\big) \\
= &E^* \big(\| \X_{\A^c}^{*T}\X_{\A}^*(\bs{\hat{\beta}}_{\A} - \bs{\hat{\beta}}_{\A}^*)\|^2 I(\tilde{\psi}>\psi/2)\big) \\ 
\le &(nD)^2 E^*\big(\|\bs{\hat{\beta}}_{\A} - \bs{\hat{\beta}}_{\A}^*\|^2 I(\tilde{\psi}>\psi/2)\big) (1+o_p(1)) \\
\le &\left({2np_0D^3\sigma^2 \over d^2} + {2\lambda_n^2p_0(\psi/2)^{-2\gamma}D^2 \over nd^2}\right) (1+o_p(1)), 
\end{align*}
we have for $B_2$, 
\begin{align*}
B_2 &\le \sum_{j \notin \A} {E^*\big(|\X_j^{*T}\X_{\A}^*(\bs{\hat{\beta}}_{\A} - \bs{\hat{\beta}}_{\A}^*)|^2I(\tilde{\psi}>\psi/2)\big) \over \kappa^2\lambda_n^2\omega_j^2} + c_1 \\
&\le \left( {2np_0D^3\sigma^2 \over \lambda_n^2O_p\left({p_n \over n}\right)^{-\gamma}\kappa^2d^2} + {2\lambda_n^2p_0(\psi/2)^{-2\gamma}D^2 \over n\lambda_n^2O_p\left({p_n \over n}\right)^{-\gamma}\kappa^2d^2} \right) (1+o_p(1)) + O_p\left({p_n \over n\psi^2}\right)\\
&\le O_p\left({n \over \lambda_n^2 n^{(1-\varrho)(1+\gamma)-1}}\right) + O_p\left({p_0 \over n}\big({p_n \over n\psi^2}\big)^\gamma\right) + O_p\left({p_n \over n\psi^2}\right) \\
&\to 0. 
\end{align*}

Hence (\ref{eq2}) is proved.  We have shown that $\bs{\hat{\beta}}^*=(\bs{\hat{\beta}}_{\A}^*, \b{0})$ and $sgn(\bs{\hat{\beta}}_{\A}^*) = sgn(\bs{\hat{\beta}}_{\A})$ with probability tending to 1, where $\bs{\hat{\beta}}^*$ is the adaptive LASSO estimate using paired bootstrap data.  Also it can be deduced from (\ref{eq23}) that \\
$P(\min_{j \in \A} |\hat{\beta}_j| > 0) \to 1$.  To sum up, we get $\lim_{n \to \infty} P^*(\mathcal{A}_n^*= \A \mid \lambda_n) = 1$.

We now prove $\lim_{n \to \infty} P^*(\mathcal{A}_n^*= \M_r \mid \lambda_n') < 1$, where $\M_r$ is any $r$-dimensional model, $p_0<r<p_n$, and $\lambda_n'$ is a tuning parameter such that the adaptive LASSO estimator under $\lambda_n'$ is of dimension $r$.  Then $\lambda_n' < \lambda_n$, hence $\lambda_n'/{\sqrt{n}} \to 0$.  If it also satisfies $\lim_{n \to \infty}{\lambda_n'^{2} n^{(1-\varrho)(1+\gamma)-1} \over n} \to \infty$,  we would have $P^*(\mathcal{A}_n^*= \A \mid \lambda_n')=1$ based on previous proof, which contradicts with the definition of $\lambda_n'$.  Therefore,  
$$\lim_{n \to \infty}{\lambda_n'^{2} n^{(1-\varrho)(1+\gamma)-1} \over n} < \infty.$$

To prove $\lim_{n \to \infty} P^*(\mathcal{A}_n^*= \M_r \mid \lambda_n') < 1$, by the KKT regularity conditions it suffices to show 
$$ P^* \left\{ \forall j \notin \M_r, |\X_j^{*T}(\y^* - \X^* \bs{\hat{\beta}}^*)| < \lambda_n' \omega_j \right\} < 1, $$
or equivalently 
\begin{align}
P^* \left\{ \exists j \notin \M_r, |\X_j^{*T}(\y^* - \X^* \bs{\hat{\beta}}^*)| \ge \lambda_n' \omega_j \right\} > 0. 
\label{eq3}
\end{align}

Following previous proof, we get 
\begin{align*}
&P^* \left\{ \exists j \notin \M_r, |\X_j^{*T}(\y^* - \X^* \bs{\hat{\beta}}^*)| \ge \lambda_n' \omega_j \right\} \\
\le &\sum_{j \notin \M_r} P^* \left\{ |\X_j^{*T}(\y^* - \X^* \bs{\hat{\beta}})| \ge (1-\kappa) \lambda_n' \omega_j \right\}  \\
&+ \sum_{j \notin \M_r} P^* \left\{ |\X_j^{*T}\X^*(\bs{\hat{\beta}} - \bs{\hat{\beta}}^*)| \ge \kappa \lambda_n' \omega_j \right\} \\
= &B_1 + B_2. 
\end{align*}
However, 
\begin{align*}
B_1 \le {np_{\M_r^c}\sigma^2 + (n-1)p_{\M_r^c}D\sigma^2 \over (1-\kappa)^2\lambda_n'^2O_p\left({p_n \over n}\right)^{-\gamma}} = O_p\left({n \over \lambda_n'^2 n^{(1-\varrho)(1+\gamma)-1}}\right) \not \to 0,   
\end{align*}
as $n \to \infty$. Similarly, $\lim_{n \to \infty}B_2 \not \to 0$.  Then (\ref{eq3}) holds.  
\end{proof}

\begin{lemma}
Suppose conditions (A1) and (A5) hold and $\lim_{n \to \infty} \lambda_{n2} / \sqrt{n} = 0$ in ridge estimates. Then, 
$$ E^* [\X^T\bs{\v}^*] = \b{0}, \quad \lim_{n \to \infty} \Var^* [\X^T\bs{\v}^*] = \X^T\X \sigma^2 \ \text{with probability 1}. $$
\label{lem2}
\end{lemma}

\begin{proof}
Assume $\bs{\hat{\beta}}$ is a ridge estimate of $\bs{\beta}_0$, 
$$ \bs{\hat{\beta}} = \arg \min_{\bs{\beta}} \|\y - \X\bs{\beta}\|^2 + \lambda_{n2}\|\bs{\beta}\|^2. $$
By (\ref{eq18}), $E\|\bs{\hat{\beta}} - \bs{\beta}_0\|^2 \le O_p\big({p_n \over n}\big)$. Calculate centered residuals $\bs{\hat{\v}}$,
$$ \bs{\hat{\v}}_0 = \y - \X\bs{\hat{\beta}}, \quad  \bs{\hat{\v}} = \bs{\hat{\v}}_0 - \bs{\bar{\v}}_0, $$
where each entry of $\bs{\bar{\v}}_0$, marked as $\bar{\v}_0$, is the mean of  $\bs{\hat{\v}}_0$.  Denote $\bs{\v}^* = (\v_1^*, \dots, \v_n^*)^T$ an i.i.d bootstrap sample from the empirical distribution that puts mass $n^{-1}$ on each entry of $\bs{\hat{\v}}$.

By definition, we have 
$$ E^* [\X^T\bs{\v}^*] = \X^T E^*(\bs{\v}^*) = \b{0}, $$
$$ \Var^* [\X^T\bs{\v}^*] = \X^T\X \Var^*(\v_1^*) = \X^T\X E^*(\v_1^{*2}), $$
and 
$$ E^*(\v_1^{*2}) = {1 \over n} \sum_{i=1}^n (\hat{\v}_{0i} - \bar{\v}_{0})^2. $$

In above equation, 
\begin{align*}
\bar{\v}_{0} = {1 \over n} \sum_{i=1}^n \hat{\v}_{0i} = {1 \over n} \sum_{i=1}^n (y_i - \x_i^T\bs{\hat{\beta}}) = {1 \over n} \sum_{i=1}^n \x_i^T(\bs{\beta}_0 - \bs{\hat{\beta}}) + {1 \over n} \sum_{i=1}^n \v_i.
\end{align*}

Moreover, by the sum of squares inequality, 
\begin{align*}
\left| {1 \over n} \sum_{i=1}^n \x_i^T(\bs{\beta}_0 - \bs{\hat{\beta}}) \right| &\le \left\{ {1 \over n} \sum_{i=1}^n \left[ \x_i^T(\bs{\beta}_0 - \bs{\hat{\beta}}) \right]^2 \right\}^{1/2} \\
&= \left\{ {1 \over n} \|\X(\bs{\beta}_0 - \bs{\hat{\beta}})\|^2 \right\}^{1/2} \\
&\le \left\{ {\zeta_{max}(\X^T\X) \over n} \|\bs{\beta}_0 - \bs{\hat{\beta}}\|^2 \right\}^{1/2} \\
&= O_p\Big(\sqrt{{p_n \over n}}\Big). 
\end{align*}

Hence, 
$$ \bar{\v}_{0} = {1 \over n} \sum_{i=1}^n \v_i +  O_p\Big(\sqrt{{p_n \over n}}\Big). $$

Let  
$$ s_n^2 = {1 \over n} \sum_{i=1}^n (\hat{\v}_{0i} - \bar{\v}_{0})^2 \quad \text{and} \quad \sigma_n^2 = {1 \over n} \sum_{i=1}^n (\v_i - \bar{\v})^2, $$
where $\bar{\v} = {1 \over n} \sum_{i=1}^n \v_i$.  We now prove $s_n \to \sigma_n$ asymptotically. 

Note that 
$$ \lim_{n \to \infty} \sigma_n^2 = \lim_{n \to \infty} {1 \over n} \sum_{i=1}^n \v_i^2 - \Big({1 \over n} \sum_{i=1}^n \v_i\Big)^2 = E(\v_i^2) - \left( E(\v_i) \right)^2 = \sigma^2 $$
with probability 1. 

And by the sum of squares inequality,

\begin{align*}
&(s_n - \sigma_n)^2 = \left\{ \left[ {1 \over n} \sum_{i=1}^n (\hat{\v}_{0i} - \bar{\v}_{0})^2 \right]^{1/2} - \left[ {1 \over n} \sum_{i=1}^n (\v_i - \bar{\v})^2 \right]^{1/2} \right\}^2 \\
&= {1 \over n} \sum_{i=1}^n (\hat{\v}_{0i} - \bar{\v}_{0})^2 + {1 \over n} \sum_{i=1}^n (\v_i - \bar{\v})^2 - 2\left[ {1 \over n} \sum_{i=1}^n (\hat{\v}_{0i} - \bar{\v}_{0})^2 \right]^{1 \over 2} \left[ {1 \over n} \sum_{i=1}^n (\v_i - \bar{\v})^2 \right]^{1 \over 2} \\
&\le {1 \over n} \sum_{i=1}^n (\hat{\v}_{0i} - \bar{\v}_{0})^2 + {1 \over n} \sum_{i=1}^n (\v_i - \bar{\v})^2 - {2 \over n}\sum_{i=1}^n (\hat{\v}_{0i} - \bar{\v}_{0})(\v_i - \bar{\v}) \\
&= {1 \over n} \sum_{i=1}^n \left[ (\hat{\v}_{0i} - \bar{\v}_{0}) - (\v_i - \bar{\v}) \right]^2 \\
&= {1 \over n} \sum_{i=1}^n \left[ \hat{\v}_{0i} - \v_i - O_p\Big(\sqrt{{p_n \over n}}\Big) \right]^2 \\
&\le {1 \over n} \|\X(\bs{\beta}_0 - \bs{\hat{\beta}})\|^2 + O_p\Big(\sqrt{{p_n \over n}}\Big) {1 \over \sqrt{n}} \|\X(\bs{\beta}_0 - \bs{\hat{\beta}})\| + O_p\Big({p_n \over n}\Big) \\
&\le {\zeta_{max}(\X^T\X) \over n} \|\bs{\beta}_0 - \bs{\hat{\beta}}\|^2 + O_p\Big(\sqrt{{p_n \over n}}\Big) \sqrt{{\zeta_{max}(\X^T\X) \over n}} \|\bs{\beta}_0 - \bs{\hat{\beta}}\| + O_p({p_n \over n}) \\
&= O_p\Big({p_n \over n}\Big). 
\end{align*}
Then $\lim_{n \to \infty} s_n^2 = \sigma^2$ with probability 1. 
\end{proof}

\begin{proof}[Proof of Theorem 2]
Let $(\X, \y^*)$ be a residual bootstrap sample, where $\y^* = \X\bs{\hat{\beta}} + \bs{\v}^*$ and $\bs{\hat{\beta}}$ is the ridge estimator.  Define  
\begin{align}
\bs{\tilde{\beta}}^* = (1+{\lambda_{n2} \over n}) \left\{ \arg \min_{\bs{\beta}} \|\y^* - \X\bs{\beta}\|^2 + \lambda_{n2}\sum_{j=1}^{p_n} |\beta_j|^2 + \lambda_{n1}^+\sum_{j=1}^{p_n} \omega_j|\beta_j| \right\},
\label{eq4}
\end{align}
where we dropped the subscript `ae' in $\bs{\tilde{\beta}}_{ae}^*$ for simplicity. 

Let 
$$ \bs{\tilde{\beta}}_{\A}^* = \arg \min_{\bs{\beta}} \|\y^* - \X_{\A}\bs{\beta}\| + \lambda_{n2} \sum_{j \in \A}|\beta_j|^2 + \lambda_{n1}^+\sum_{j \in \A} \omega_j|\beta_j|, $$
we prove $((1+{\lambda_{n2} \over n})\bs{\tilde{\beta}}_{\A}^*, \b{0})$ is the solution to (\ref{eq4}) with probability tending to 1.  By the KKT regularity conditions, this suffices to show 
$$ P^* \left\{ \forall j \notin \A, |\X_j^{T}(\y^* - \X_{\A} \bs{\tilde{\beta}}_{\A}^*)| < \lambda_{n1}^+ \omega_j \right\} \to 1, $$
or equivalently 
\begin{align}
P^* \left\{ \exists j \notin \A, |\X_j^{T}(\y^* - \X_{\A} \bs{\tilde{\beta}}_{\A}^*)| \ge \lambda_{n1}^+ \omega_j \right\} \to 0. 
\label{eq5}
\end{align}

Note that $\omega_j = |\hat{\beta}_{ej}|^{-\gamma}$ where $\bs{\hat{\beta}}_e = (\hat{\beta}_{e1}, \dots, \hat{\beta}_{ep_n})^T$ is the Elastic-Net estimator defined in (\ref{eqt4}).  By Theorem 3.1 in \cite{zou2009adaptive},  
\begin{align}
E\|\bs{\hat{\beta}}_e - \bs{\beta}_0\|^2 \le 4{\lambda_{n2}^2\|\bs{\beta}_0\|^2 + np_nD\sigma^2 + \lambda_{n1}^2p_n \over (nd+\lambda_{n2})^2} = O_p\left({p_n \over n}\right)
\label{eq19}
\end{align}
under condition (A4).

Let $\psi = \min_{j \in \A} |\beta_{0j}|$ and $\tilde{\psi} = \min_{j \in \A} |\hat{\beta}_{ej}|$.  Then 
\begin{align*}
&P^* \left\{ \exists j \notin \A, |\X_j^{T}(\y^* - \X_{\A} \bs{\tilde{\beta}}_{\A}^*)| \ge \lambda_{n1}^+ \omega_j \right\} \\
\le &P^* \left\{ \exists j \notin \A, |\X_j^{T}(\y^* - \X_{\A} \bs{\tilde{\beta}}_{\A}^*)| \ge \lambda_{n1}^+ \omega_j,  \tilde{\psi}>\psi/2 \right\} + P\{\tilde{\psi} \le \psi/2\} \\
\le &\sum_{j \notin \A} P^* \left\{ |\X_j^{T}(\y^* - \X_{\A} \bs{\tilde{\beta}}_{\A}^*)| \ge \lambda_{n1}^+ \omega_j,  \tilde{\psi}>\psi/2 \right\} + P\{\tilde{\psi} \le \psi/2\} \\
= &B_1 + B_2.
\end{align*}

By (\ref{eq19}) under condition (A4),
\begin{align*}
B_2 &= P\{\tilde{\psi} \le \psi/2\} \le P\{\|\bs{\hat{\beta}}_e - \bs{\beta}_0\| \ge \psi/2\} \\
&\le {4E\|\bs{\hat{\beta}}_e - \bs{\beta}_0\|^2 \over \psi^2} \le O_p({p_n \over n\psi^2}) \to 0. 
\end{align*}

Also by (\ref{eq19}) $\forall j \in \A^c$, $E|\hat{\beta}_{ej}|^2 \le E\|\bs{\hat{\beta}}_e - \bs{\beta}_0\|^2 = O_p\left({p_n \over n}\right)$, which indicates $|\hat{\beta}_{ej}| \le O_p\left({p_n \over n}\right)^{1/2}$.  Hence 
$$ B_1 \le {O_p\left({p_n \over n}\right)^\gamma \over \lambda_{n1}^{+2}} E^*\bigg\{ \sum_{j \notin \A}|\X_j^{T}(\y^* - \X_{\A} \bs{\tilde{\beta}}_{\A}^*)|^2 I(\tilde{\psi}>\psi/2) \bigg\}. $$

Note that 
\begin{align*}
&E^*\bigg\{\sum_{j \notin \A}|\X_j^{T}(\y^* - \X_{\A} \bs{\tilde{\beta}}_{\A}^*)|^2\bigg\} \\
= &E^*\bigg\{\sum_{j \notin \A} |\X_j^{T}(\X_{\A} \bs{\hat{\beta}}_{\A} + \X_{\A^c} \bs{\hat{\beta}}_{\A^c} + \bs{\v}^* - \X_{\A} \bs{\tilde{\beta}}_{\A}^*)|^2\bigg\} \\
\le &3E^*\|\X_{\A^c}^T\X_{\A}(\bs{\hat{\beta}}_{\A} - \bs{\tilde{\beta}}_{\A}^*)\|^2 + 3\|\X_{\A^c}^T\X_{\A^c}\bs{\hat{\beta}}_{\A^c}\|^2 + 3E^*\|\X_{\A^c}^T\bs{\v}^*\|^2 \\
\le &3(nD)^2E^*\|\bs{\hat{\beta}}_{\A} - \bs{\tilde{\beta}}_{\A}^*\|^2 + 3(nD)^2\|\bs{\hat{\beta}}_{\A^c}\|^2 + 3E^*\|\X_{\A^c}^T\bs{\v}^*\|^2. 
\end{align*}

By (\ref{eq18}),  
\begin{align}
\|\bs{\hat{\beta}}_{\A^c}\|^2 \le \|\bs{\hat{\beta}} - \bs{\beta}_0\|^2 \le O_p({p_n \over n}). 
\label{eq6}
\end{align}

We now study $E^*\|\bs{\hat{\beta}}_{\A} - \bs{\tilde{\beta}}_{\A}^*\|^2$.  Let
$$ \bs{\tilde{\beta}}_{\A}^*(\lambda_{n2}, 0) = \arg\min_{\bs{\beta}} \|\y^* - \X_{\A}\bs{\beta}\| + \lambda_{n2}\sum_{j \in \A}\beta_j^2. $$

By using the same arguments for deriving (6.3) in \cite{zou2009adaptive}, we can easily show
\begin{align}
\|\bs{\tilde{\beta}}_{\A}^* - \bs{\tilde{\beta}}_{\A}^*(\lambda_{n2}, 0)\| \le {\lambda_{n1}^+\|\bs{\omega}_{\A}\| \over \zeta_{min}(\X_{\A}^T\X_{\A}) + \lambda_{n2}}. 
\label{eq7}
\end{align}

On the other hand, 
\begin{align*}
\bs{\tilde{\beta}}_{\A}^*(\lambda_{n2}, 0) - \bs{\hat{\beta}}_{\A} = (\X_{\A}^T\X_{\A} + \lambda_{n2}I)^{-1} (-\lambda_{n2}\bs{\hat{\beta}}_{\A} + \X_{\A}^T\X_{\A^c}\bs{\hat{\beta}}_{\A^c} + \X_{\A}^T\bs{\v}^*), 
\end{align*}
by Lemma \ref{lem2},
\begin{align}
E^*\|\bs{\tilde{\beta}}_{\A}^*(\lambda_{n2}, 0) - \bs{\hat{\beta}}_{\A}\|^2 &\le 3{\lambda_{n2}^2\|\bs{\hat{\beta}}_{\A}\|^2 + \|\X_{\A}^T\X_{\A^c}\bs{\hat{\beta}}_{\A^c}\|^2 + E^*\|\X_{\A}^T\bs{\v}^*\|^2 \over \left(\zeta_{min}(\X_{\A}^T\X_{\A}) + \lambda_{n2}\right)^2} \nonumber \\
&\le 3{\lambda_{n2}^2\|\bs{\hat{\beta}}_{\A}\|^2 + (nD)^2\|\bs{\hat{\beta}}_{\A^c}\|^2 + np_0D\sigma^2 \over (nd+\lambda_{n2})^2}. 
\label{eq8}
\end{align}

By assembling (\ref{eq6})--(\ref{eq8}), we get
\begin{align*}
E^*\|\bs{\hat{\beta}}_{\A} - \bs{\tilde{\beta}}_{\A}^*\|^2 \le 2E^*\|\bs{\tilde{\beta}}_{\A}^* - \bs{\tilde{\beta}}_{\A}^*(\lambda_{n2}, 0)\|^2 + 2E^*\|\bs{\tilde{\beta}}_{\A}^*(\lambda_{n2}, 0) - \bs{\hat{\beta}}_{\A}\|^2 \\
\le 6{\lambda_{n1}^{+2}\|\bs{\omega}_{\A}\|^2 + \lambda_{n2}^2\|\bs{\hat{\beta}}_{\A}\|^2 + O_p(np_nD^2) + np_0D\sigma^2 \over (nd+\lambda_{n2})^2}.
\end{align*}

And
\begin{align*}
&E^*\left\{ \sum_{j \notin \A}|\X_j^{T}(\y^* - \X_{\A} \bs{\tilde{\beta}}_{\A}^*)|^2 I(\tilde{\psi}>\psi/2) \right\} \le 3O_p(np_nD^2) + 3np_{\A^c}D\sigma^2 \\
&+ 18n^2D^2{\lambda_{n1}^{+2}p_0(\psi/2)^{-2\gamma} +  \lambda_{n2}^2\|\bs{\hat{\beta}}_{\A}\|^2 + O_p(np_nD^2) + np_0D\sigma^2 \over (nd+\lambda_{n2})^2} \\
&= O_p(np_n) + O_p(\psi^{-2\gamma}\lambda_{n1}^{+2}p_0).  
\end{align*}

Then under conditions (A1)--(A2) and (A4)--(A5), 
\begin{align*}
B_1 &\le {O_p\left({p_n \over n}\right)^\gamma \over \lambda_{n1}^{+2}} [O_p(np_n) + O_p(\psi^{-2\gamma}\lambda_{n1}^{+2}p_0)] \\
&\le O_p\left({n \over \lambda_{n1}^{+2} n^{(1-\varrho)(1+\gamma)-1}}\right) + O_p\left({1 \over \psi^{2\gamma} n^{(1-\varrho)(1+\gamma)-1}}\right) \\
&\to 0. 
\end{align*}

Hence (\ref{eq5}) is proved.  So far we have shown that $\bs{\tilde{\beta}}^* = ((1+{\lambda_{n2} \over n})\bs{\tilde{\beta}}_{\A}^*, \b{0})$  with probability tending to 1, where $\bs{\tilde{\beta}}^*$ is the adaptive Elastic-Net estimate using residual bootstrap data.   To prove $\lim_{n \to \infty} P^*(\mathcal{T}_n^*= \A \mid \lambda_{n1}^+ ) = 1$, we still need to show that $P(\min_{j \in \A} |\tilde{\beta}_j^*| > 0) \to 1$.

Let $\hat{\psi} = \min_{j \in \A} |\hat{\beta}_j|$ and $\tilde{\psi}^* = \min_{j \in \A} |\tilde{\beta}_j^*|$.  By (\ref{eq18}), 
$$ P(\hat{\psi} \le \psi/2) \le P(\|\bs{\hat{\beta}} - \bs{\beta}_0\| \ge \psi/2) \le O_p\big({p_n \over n\psi^2}\big) \to 0. $$
Hence $P(\hat{\psi} > \psi/2) \to 1$ as $n \to \infty$ where $\psi > 0$.  Under condition (A4), 
\begin{align*}
P(\tilde{\psi}^* \le \hat{\psi}/2) &\le P(\tilde{\psi}^* \le \hat{\psi}/2, \tilde{\psi} > \psi/2) + P(\tilde{\psi} \le \psi/2) \\
&\le P(\|\bs{\hat{\beta}}_{\A} - \bs{\tilde{\beta}}_{\A}^*\| \ge \hat{\psi}/2, \tilde{\psi} > \psi/2) + B_2  \\
&\le {16 \over \psi^2} E^*\big(\|\bs{\hat{\beta}}_{\A} - \bs{\tilde{\beta}}_{\A}^*\|^2I(\tilde{\psi} > \psi/2)\big) + B_2 \\
&\le {96 \over \psi^2} {\lambda_{n1}^{+2}p_0(\psi/2)^{-2\gamma} +  \lambda_{n2}^2\|\bs{\hat{\beta}}_{\A}\|^2 + O_p(np_nD^2) + np_0D\sigma^2 \over (nd+\lambda_{n2})^2} + B_2 \\
&= O_p\bigg(\Big({\lambda_{n1}^+ \over \sqrt{n}\psi^\gamma}\Big)^2{p_0 \over n\psi^2}\bigg) + O_p\big({p_n \over n\psi^2}\big) \\
&\to 0, 
\end{align*}
which indicates $P(\tilde{\psi}^* > \hat{\psi}/2) \to 1$ as $n \to \infty$.  To sum up, $\lim_{n \to \infty} P(\tilde{\psi}^* > \psi/4) = 1$.  Thus $\lim_{n \to \infty} P^*(\mathcal{T}_n^*= \A \mid \lambda_{n1}^+ ) = 1$ is proved.

We now prove $\lim_{n \to \infty} P^*(\mathcal{T}_n^*= \M_r \mid \lambda_{n1}' ) < 1$, where $\M_r$ is any $r$-dimensional model, $p_0<r<p_n$, and $\lambda_{n1}'$ is a tuning parameter such that the adaptive Elastic-Net estimator under $\lambda_{n1}'$ is of dimension $r$.  Then $\lambda_{n1}' < \lambda_{n1}^+$, hence $\lambda_{n1}'/\sqrt{n} \to 0$. If it also satisfies $\lim_{n \to \infty} {\lambda_{n1}'^2 n^{(1-\varrho)(1+\gamma)-1} \over n} \to \infty$, we would have $P^*(\mathcal{T}_n^*= \A \mid \lambda_{n1}')=1$ based on previous proof, which contradicts with the definition of $\lambda_{n1}'$.  Therefore,  
$$\lim_{n \to \infty} {\lambda_{n1}'^2 n^{(1-\varrho)(1+\gamma)-1} \over n} < \infty.$$

To prove $\lim_{n \to \infty} P^*(\mathcal{T}_n^*= \M_r \mid \lambda_{n1}' ) < 1$, by the KKT regularity conditions it suffices to show  
$$ P^* \left\{ \forall j \notin \M_r, |\X_j^{*T}(\y^* - \X_{\M_r} \bs{\hat{\beta}}_{\M_r}^*)| < \lambda_{n1}' \omega_j \right\} < 1, $$
or equivalently 
\begin{align*}
P^* \left\{ \exists j \notin \M_r, |\X_j^{*T}(\y^* - \X_{\M_r} \bs{\hat{\beta}}_{\M_r}^*)| \ge \lambda_{n1}' \omega_j \right\} > 0. 
\end{align*}

By following the same arguments for showing (\ref{eq5}), we get
\begin{align*}
&P^* \left\{ \exists j \notin \M_r, |\X_j^{*T}(\y^* - \X_{\M_r} \bs{\hat{\beta}}_{\M_r}^*)| \ge \lambda_{n1}' \omega_j \right\} \\
&\le  {O_p\left({p_n \over n}\right)^\gamma \over \lambda_{n1}'^2} \bigg\{3O_p(np_nD^2) + 3np_{\M_r^c}D\sigma^2 \\
&+18n^2D^2{\lambda_{n1}'^2\|\bs{\omega}_{\M_r}\|^2 +  \lambda_{n2}^2\|\bs{\hat{\beta}}_{\M_r}\|^2 + O_p(np_nD^2) + np_{\M_r}D\sigma^2 \over (nd+\lambda_{n2})^2}\bigg\} \\
&= O_p\left({n \over \lambda_{n1}'^2 n^{(1-\varrho)(1+\gamma)-1}}\right) + O_p\left(\|\bs{\omega}_{\M_r}\|^2\Big({p_n \over n}\Big)^\gamma\right) \\
&\not \to 0. 
\end{align*}

\end{proof}

\begin{lemma}
Suppose conditions (A1), (A5) and (A6) hold.  Denote $\alpha$ an overfit model including the true model, the adaptive Elastic-Net estimate $\bs{\hat{\beta}}_{s_i^c, \alpha}$ from the multi-fold CV then satisfies 
\begin{align*}
E\|\bs{\hat{\beta}}_{s_i^c, \alpha} - \bs{\beta}_{0\alpha}\|^2 &\le 4{\lambda_{n2}^2\|\bs{\beta}_{0\alpha}\|^2 + (n-t)p_{\alpha}D\sigma^2(1+o_p(1)) + \lambda_{n1}^{'2}E\|\bs{\omega}_{\alpha}\|^2 \over [(n-t)d(1+o_p(1))+\lambda_{n2}]^2} \\
&= O_p({p_{\alpha} \over n}), 
\end{align*}
where the adaptive LASSO estimate is a special case with $\lambda_{n2}=0$. 
\label{lem3}
\end{lemma}

\begin{proof}
Here we provide a proof for the adaptive LASSO estimator. The adaptive Elastic-Net estimator can be proved by using the same arguments for deriving Theorem 3.1 in \cite{zou2009adaptive} and the strategies in below. 

The adaptive LASSO estimator from the multi-fold CV is 
$$ \bs{\hat{\beta}}_{s_i^c, \alpha} = \arg\min_{\bs{\beta}} \|\b{Y}_{s_i^c} - \X_{s_i^c, \alpha}\bs{\beta}\|^2 + 2\lambda_{n1}'\sum_{j \in \alpha} \omega_j|\beta_j|, $$
which satisfies
$$ \bs{\hat{\beta}}_{s_i^c, \alpha} - \bs{\beta}_{0\alpha} = (\X_{s_i^c, \alpha}^T\X_{s_i^c, \alpha})^{-1} \left(\X_{s_i^c, \alpha}^T\bs{\v}_{s_i^c} - \lambda_{n1}'\bs{\omega}_{\alpha}\otimes sgn(\bs{\hat{\beta}}_{s_i^c, \alpha})\right). $$ 

Hence,
\begin{align*}
E\|\bs{\hat{\beta}}_{s_i^c, \alpha} - \bs{\beta}_{0\alpha}\|^2 &\le {2E\|\X_{s_i^c, \alpha}^T\bs{\v}_{s_i^c}\|^2 + 2\lambda_{n1}^{'2}E\|\bs{\omega}_{\alpha}\|^2 \over \zeta_{min}^2(\X_{s_i^c, \alpha}^T\X_{s_i^c, \alpha})} \\
&\le {2\zeta_{max}(\X_{s_i^c, \alpha}^T\X_{s_i^c, \alpha})p_{\alpha}\sigma^2 + 2\lambda_{n1}^{'2}E\|\bs{\omega}_{\alpha}\|^2 \over \zeta_{min}^2(\X_{s_i^c, \alpha}^T\X_{s_i^c, \alpha})} \\
&\le {2(n-t)p_{\alpha}D\sigma^2(1+o_p(1)) + 2\lambda_{n1}^{'2}E\|\bs{\omega}_{\alpha}\|^2 \over (n-t)^2d^2(1+o_p(1))} \\
&= O_p\left({p_{\alpha} \over n}\right).
\end{align*}
The last equation holds because $\lambda_{n1}'$ continuously decreases from $\lambda_{n1}^+$ to 0 as $\alpha$ changes from the true model to full model. 
\end{proof}

\begin{proof}[Proof of Theorem 3]
We integrate the proof for adaptive Elastic-Net and adaptive LASSO.  Denote $\a$ an overfit model including the true model.  The $\MCV_{\a}$ is 
\begin{align}
\MCV_{\a} &= {1 \over n} \sum_{i=1}^K \|\X_{s_i,\a}\bs{\beta}_{0\a} + \bs{\v}_{s_i} - \X_{s_i, \a} \bs{\hat{\beta}}_{s_i^c, \a} \|^2 \nonumber \\
&= {1 \over n} \bs{\v}^T \bs{\v} + {1 \over n} \sum_{i=1}^K \| \X_{s_i, \a}(\bs{\beta}_{0\a} - \bs{\hat{\beta}}_{s_i^c, \a}) \|^2 \nonumber \\
&+ {2 \over n} \sum_{i=1}^K (\bs \beta_{0\a} - \bs{\hat{\beta}}_{s_i^c, \a})^T \X_{s_i, \alpha}^T \bs{\v}_{s_i}.
\label{eq10}
\end{align}

By Lemma \ref{lem3}, the second term in (\ref{eq10}) satisfies
\begin{align*}
E\| \X_{s_i, \a}(\bs{\beta}_{0\a} - \bs{\hat{\beta}}_{s_i^c, \a}) \|^2 &\le \zeta_{max}(\X_{s_i, \a}^T\X_{s_i, \a}) E\|\bs{\beta}_{0\a} - \bs{\hat{\beta}}_{s_i^c, \a}\|^2 \\
&\le tDO_p\left({p_{\a} \over n}\right) = O_p\left({tp_{\a} \over n}\right), \\
\Var\| \X_{s_i, \a}(\bs{\beta}_{0\a} - \bs{\hat{\beta}}_{s_i^c, \a}) \|^2 &\le E\| \X_{s_i, \a}(\bs{\beta}_{0\a} - \bs{\hat{\beta}}_{s_i^c, \a}) \|^4 \\
&\le O_p\left({tp_{\a} \over n}\right)^2. 
\end{align*}

Hence,
\begin{align}
\| \X_{s_i, \a}(\bs{\beta}_{0\a} - \bs{\hat{\beta}}_{s_i^c, \a}) \|^2 &\le  O_p\left({tp_{\a} \over n}\right), \nonumber \\
{1 \over n} \sum_{i=1}^K \| \X_{s_i, \a}(\bs{\beta}_{0\a} - \bs{\hat{\beta}}_{s_i^c, \a}) \|^2 &\le {K \over n} O_p\left({tp_{\a} \over n}\right) = O_p\left({p_{\a} \over n}\right). 
\label{eq11}
\end{align}

The third term in (\ref{eq10}) fulfills
\begin{align*}
E[(\bs \beta_{0\a} - \bs{\hat{\beta}}_{s_i^c, \a})^T \X_{s_i, \alpha}^T \bs{\v}_{s_i}] &= 0, \\
E|(\bs \beta_{0\a} - \bs{\hat{\beta}}_{s_i^c, \a})^T \X_{s_i, \alpha}^T \bs{\v}_{s_i}|^2 &\le E\|\bs \beta_{0\a} - \bs{\hat{\beta}}_{s_i^c, \a}\|^2 E\|\X_{s_i, \alpha}^T\bs{\v}_{s_i}\|^2 \\
&\le O_p\left({p_{\a} \over n}\right) tp_{\a}D\sigma^2 \\
&= O_p\left({tp_{\a}^2 \over n}\right) . 
\end{align*}

Hence,
\begin{align}
(\bs \beta_{0\a} - \bs{\hat{\beta}}_{s_i^c, \a})^T \X_{s_i, \alpha}^T \bs{\v}_{s_i} &\le O_p\left(\sqrt{{tp_{\a}^2 \over n}}\right), \nonumber \\
{2 \over n} \sum_{i=1}^K (\bs \beta_{0\a} - \bs{\hat{\beta}}_{s_i^c, \a})^T \X_{s_i, \alpha}^T \bs{\v}_{s_i} &\le {2K \over n} O_p\left(\sqrt{{tp_{\a}^2 \over n}}\right) = O_p\left({p_{\a} \over n}\right). 
\label{eq12}
\end{align}

By substituting (\ref{eq11})--(\ref{eq12}) to (\ref{eq10}), we obtain
$$ \MCV_{\a} = {1 \over n} \bs{\v}^T \bs{\v} + O_p\left({p_{\a} \over n}\right). $$

Let $\alpha$ and $\alpha'$ be two overfit models including the true model, then
$$ \lim_{n \to \infty} |\MCV_{\a} - \MCV_{\a'}| = \lim_{n \to \infty} \Big|O_p\Big({p_\alpha - p_{\alpha'} \over n}\Big)\Big| = 0. $$

We now consider an underfit model $\nu$. The $\MCV_{\nu}$ is
\begin{align}
\MCV_{\nu} &= {1 \over n} \sum_{i=1}^K \|\X_{s_i}\bs{\beta}_0 + \bs{\v}_{s_i} - \X_{s_i, \nu} \bs{\hat{\beta}}_{s_i^c, \nu} \|^2 \nonumber \\
&= {1 \over n} \bs{\v}^T \bs{\v} + {1 \over n} \sum_{i=1}^K \| \X_{s_i}\bs{\beta}_0 - \X_{s_i,\nu}\bs{\hat{\beta}}_{s_i^c, \nu} \|^2 \nonumber \\
&+ {2 \over n} \sum_{i=1}^K (\X_{s_i}\bs{\beta}_0 - \X_{s_i,\nu}\bs{\hat{\beta}}_{s_i^c, \nu})^T \bs{\v}_{s_i}.
\label{eq13}
\end{align}

Let $\bs{\hat{\beta}}_\nu$ be an adaptive LASSO or adaptive Elastic-Net estimator under $\nu$. The second term in (\ref{eq13}) satisfies
\begin{align}
&{1 \over n}\sum_{i=1}^k \|\X_{s_i} \bs{\beta}_0 - \X_{s_i, \nu} \bs{\hat{\beta}}_{s_i^c, \nu} \|^2 \nonumber \\  
\ge &{1 \over 2n}\sum_{i=1}^k  \|\X_{s_i}[\bs{\beta}_0 - 
\big(\begin{smallmatrix}
\bs{\hat{\beta}}_\nu \\
\b{0}_{\nu^c}
\end{smallmatrix}\big)
]\|^2 - {1 \over n}\sum_{i=1}^k  \|\X_{s_i, \nu}(\bs{\hat{\beta}}_{s_i^c, \nu} - \bs{\hat{\beta}}_\nu)\|^2 \nonumber \\
\ge &{1 \over 2n}\sum_{i=1}^k \zeta_{min}(\X_{s_i}^T\X_{s_i})\|\bs{\beta}_0 - \big(\begin{smallmatrix}
\bs{\hat{\beta}}_\nu \\
\b{0}_{\nu^c}
\end{smallmatrix}\big)
\|^2 - {1 \over n}\sum_{i=1}^k \zeta_{max}(\X_{s_i, \nu}^T \X_{s_i, \nu}) \|\bs{\hat{\beta}}_{s_i^c, \nu} - \bs{\hat{\beta}}_\nu\|^2 \nonumber \\
\ge &{d\|\bs{\beta}_{0\nu^c}\|^2 \over 2} - o_p(1). 
\label{eq14}
\end{align}

For the third term in (\ref{eq13}),
\begin{align*}
&E[(\X_{s_i}\bs{\beta}_0 - \X_{s_i,\nu}\bs{\hat{\beta}}_{s_i^c, \nu})^T \bs{\v}_{s_i}] = 0, \\
&E|(\X_{s_i}\bs{\beta}_0 - \X_{s_i,\nu}\bs{\hat{\beta}}_{s_i^c, \nu})^T \bs{\v}_{s_i}|^2 \\
\le &2E|[\bs{\beta}_0 - 
\big(\begin{smallmatrix}
\bs{\hat{\beta}}_\nu \\
\b{0}_{\nu^c}
\end{smallmatrix}\big)
]^T\X_{s_i}^T\bs{\v}_{s_i}|^2 +2E|(\bs{\hat{\beta}}_\nu - \bs{\hat{\beta}}_{s_i^c, \nu})^T\X_{s_i, \nu}^T\bs{\v}_{s_i}|^2 \\
\le &2\left(\|\bs{\beta}_{0\nu^c}\|^2 + o_p(1)\right) E\|\X_{s_i}^T\bs{\v}_{s_i}\|^2 + 2o_p(1)E\|\X_{s_i, \nu}^T\bs{\v}_{s_i}\|^2 \\
\le &2\left(\|\bs{\beta}_{0\nu^c}\|^2 + o_p(1)\right)tp_nD\sigma^2 + 2tp_{\nu}D\sigma^2o_p(1) \\
= &O_p\left(\|\bs{\beta}_{0\nu^c}\|^2tp_n\right). 
\end{align*}

Hence, 
\begin{align}
(\X_{s_i}\bs{\beta}_0 - \X_{s_i,\nu}\bs{\hat{\beta}}_{s_i^c, \nu})^T \bs{\v}_{s_i} &\le O_p\left(\|\bs{\beta}_{0\nu^c}\|\sqrt{tp_n}\right), \nonumber \\
{2 \over n} \sum_{i=1}^K (\X_{s_i}\bs{\beta}_0 - \X_{s_i,\nu}\bs{\hat{\beta}}_{s_i^c, \nu})^T \bs{\v}_{s_i} &\le {2K \over n}O_p\left(\|\bs{\beta}_{0\nu^c}\|\sqrt{tp_n}\right) \nonumber\\
&= O_p \left( \|\bs{\beta}_{0\nu^c}\| \sqrt{{p_n \over n}} \right). 
\label{eq15}
\end{align}

By substituting (\ref{eq14})--(\ref{eq15}) to (\ref{eq13}), we get
$$ \MCV_{\nu} \ge {1 \over n} \bs{\v}^T \bs{\v} + {d\|\bs{\beta}_{0\nu^c}\|^2 \over 2} + O_p \left(\|\bs{\beta}_{0\nu^c}\| \sqrt{{p_n \over n}}\right). $$

If $\alpha$ is an overfit model and $\nu$ is an underfit model, we have 
\begin{align}
&\lim_{n \to \infty} \MCV_{\nu} - \MCV_{\alpha} \nonumber \\
\ge &{d\|\bs{\beta}_{0\nu^c}\|^2 \over 2} + O_p \left(\|\bs{\beta}_{0\nu^c}\| \sqrt{{p_n \over n}}\right) - O_p\left({p_{\a} \over n}\right) > 0.
\end{align}

So the first part is proved.  We then combine it with Corollaries \ref{cor1}--\ref{cor4}.  For any $r$, $p_0 < r < p_n$, 
\begin{align}
\lim_{n \to \infty}{\WMF_{p_0} \over \WMF_r} = &\lim_{n \to \infty}{P^*(\mathcal{A}_n^* = \mathcal{A} \mid p_0) \exp[-\MCV_{\mathcal{A}}/{c\sigma^2}]  \over P^*(\mathcal{A}_n^* = M_r \mid r) \exp[-\MCV_{M_r}/{c\sigma^2}] } \nonumber \\
= &\lim_{n \to \infty}{P^*(\mathcal{A}_n^* = \mathcal{A} \mid p_0) \over P^*(\mathcal{A}_n^* = M_r \mid r)} \exp\Big[{\MCV_{M_r} - \MCV_{\mathcal{A}} \over c\sigma^2}\Big] \nonumber \\
= &\lim_{n \to \infty}{P^*(\mathcal{A}_n^* = \mathcal{A} \mid p_0) \over P^*(\mathcal{A}_n^* = M_r \mid r)} \exp\Big[O_p\big({r-p_0 \over n}\big)\Big] \nonumber \\
> &1.
\label{eq16}
\end{align}
And for any $r'$, $0 < r' <p_0$,
\begin{align}
&\lim_{n \to \infty}{\WMF_{p_0} \over \WMF_{r'}} = \lim_{n \to \infty}{P^*(\mathcal{A}_n^* = \mathcal{A} \mid p_0) \exp[-\MCV_{\mathcal{A}}/{c\sigma^2}]  \over P^*(\mathcal{A}_n^* = M_{r'} \mid r') \exp[-\MCV_{M_{r'}}/{c\sigma^2}] } \nonumber \\
= &\lim_{n \to \infty}{P^*(\mathcal{A}_n^* = \mathcal{A} \mid p_0) \over P^*(\mathcal{A}_n^* = M_{r'} \mid r')} \exp\Big[{\MCV_{M_{r'}} - \MCV_{\mathcal{A}} \over c\sigma^2}\Big] \nonumber \\
\ge &\lim_{n \to \infty}{P^*(\mathcal{A}_n^* = \mathcal{A} \mid p_0) \over P^*(\mathcal{A}_n^* = M_{r'} \mid r')} \exp\Big[{{d \over 2}\|\bs{\beta}_{0M_{r'}^c}\|^2 + O_p \big(\|\bs{\beta}_{0 M_{r'}^c}\| \sqrt{{p_n \over n}}\big) - O_p\big({p_0 \over n}\big) \over c\sigma^2}\Big] \nonumber \\
> &1.
\label{eq17}
\end{align} 
Then model selection consistency of the WMF procedure can be deduced from (\ref{eq16})--(\ref{eq17}).  
\end{proof}

\section{Additional simulation results}
\label{sec:asr}

\begin{figure}
\centering
\includegraphics[scale=0.42]{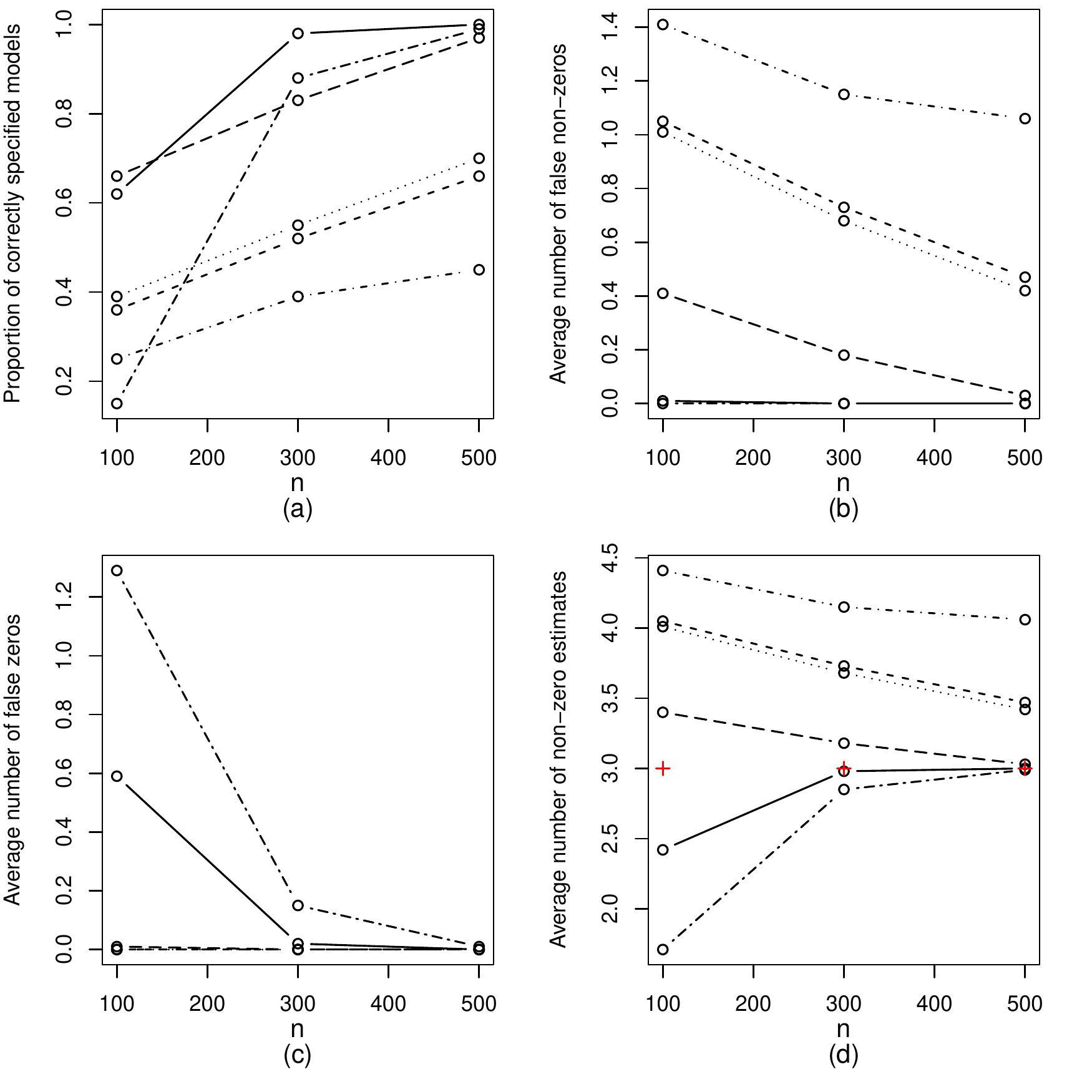}
\caption{Results of scenario 1 using residual bootstrap data: (a) proportion of correctly specified models; (b) average number of false non-zeros; (c) average number of false zeros; (d) average value of estimated model sizes.}
\end{figure}

\begin{figure}
\centering
\includegraphics[scale=0.42]{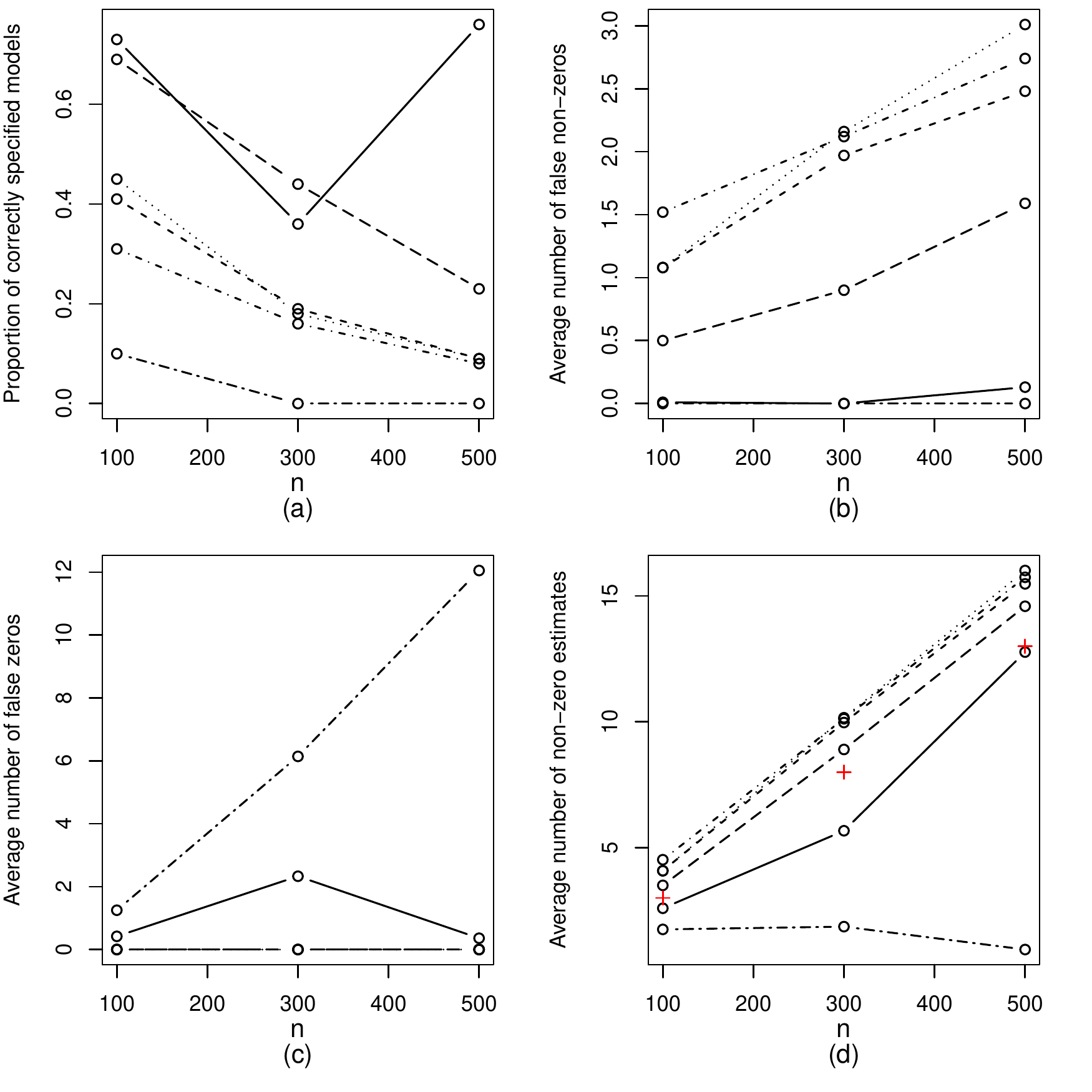}
\caption{Results of scenario 2 using residual bootstrap data: (a) proportion of correctly specified models; (b) average number of false non-zeros; (c) average number of false zeros; (d) average value of estimated model sizes.}
\end{figure}

\begin{figure}
\centering
\includegraphics[scale=0.42]{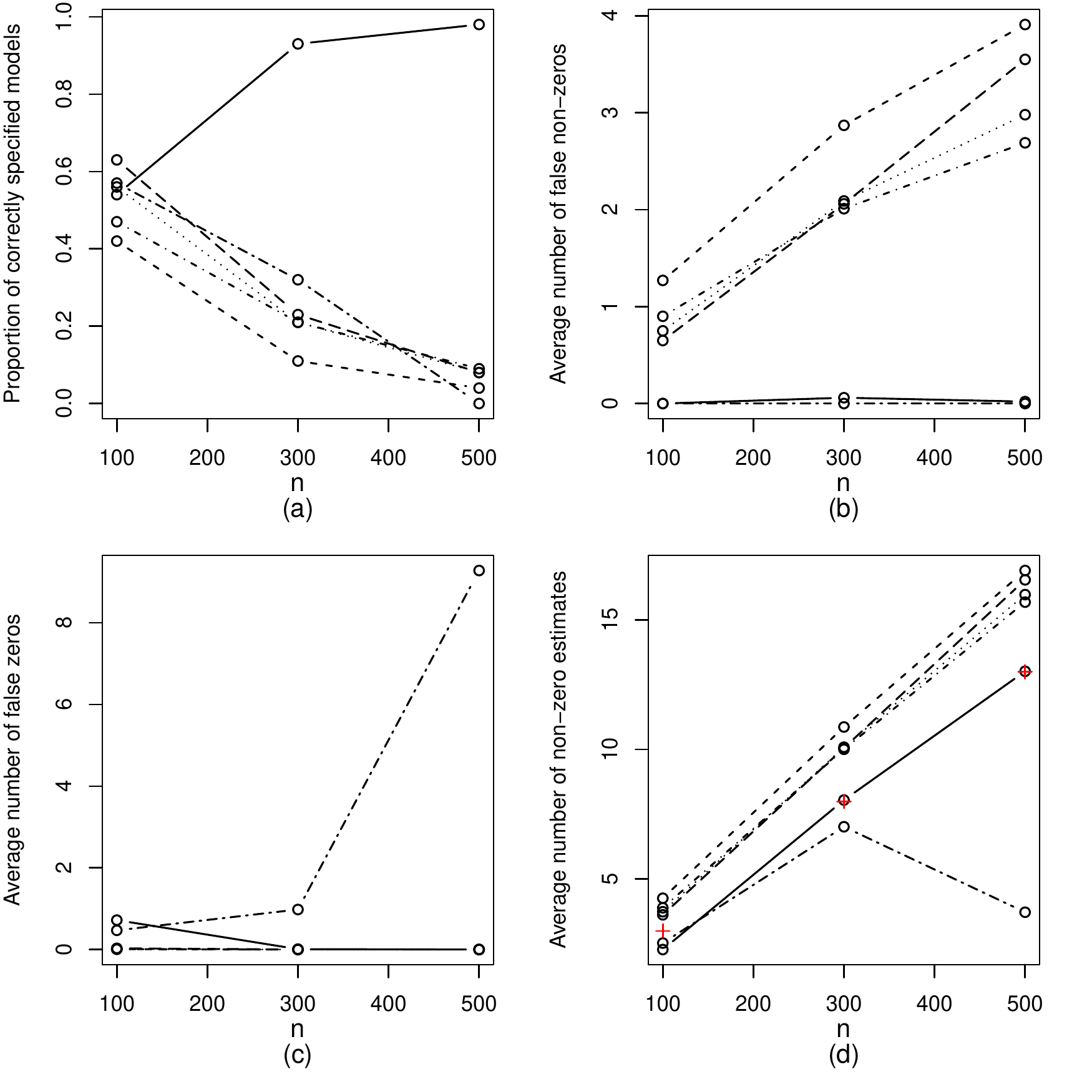}
\caption{Results of scenario 3 using residual bootstrap data: (a) proportion of correctly specified models; (b) average number of false non-zeros; (c) average number of false zeros; (d) average value of estimated model sizes.}
\end{figure}

\begin{figure}
\centering
\includegraphics[scale=0.42]{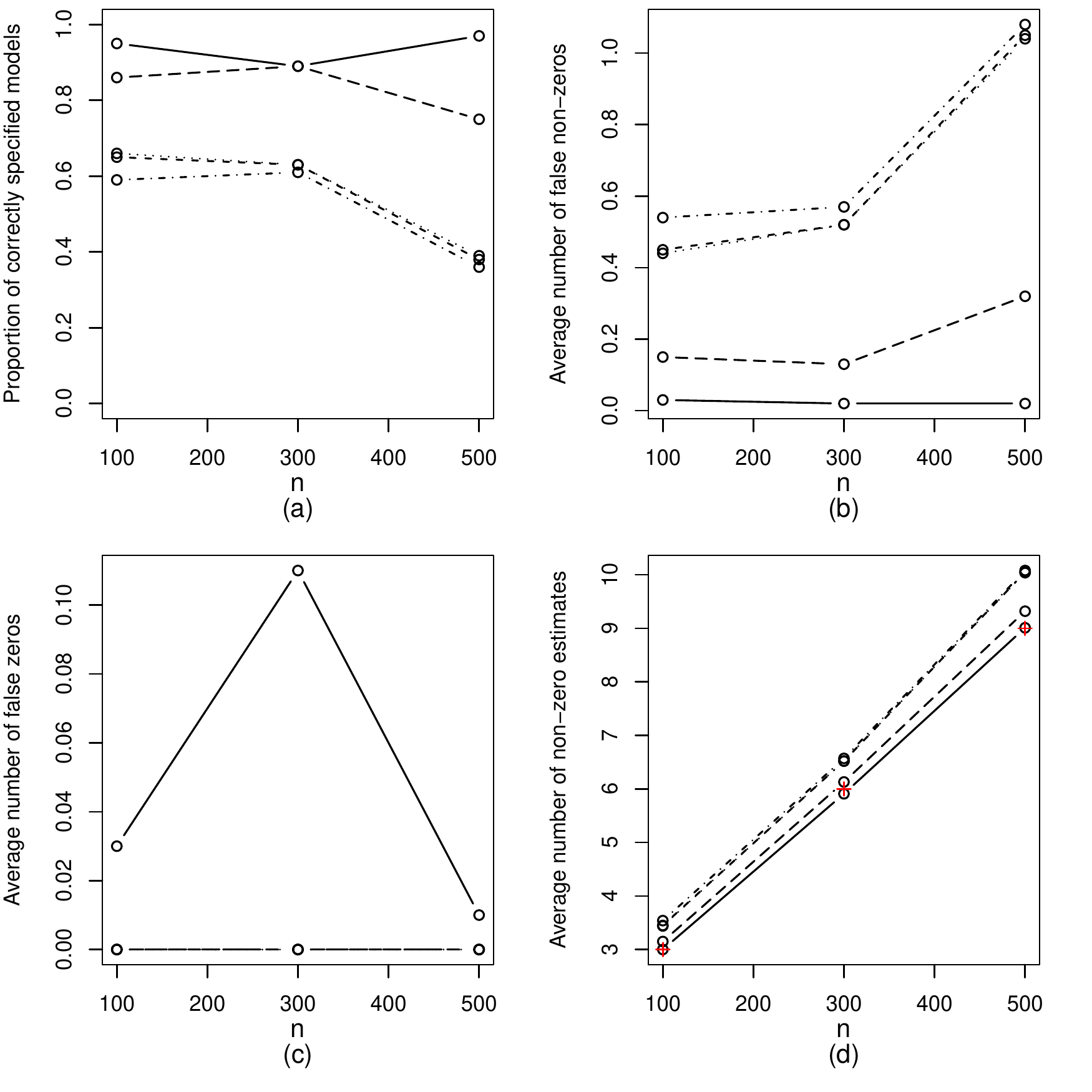}
\caption{Results of scenario 4 using paired bootstrap data: (a) proportion of correctly specified models; (b) average number of false non-zeros; (c) average number of false zeros; (d) average value of estimated model sizes.}
\end{figure}

\begin{figure}
\centering
\includegraphics[scale=0.42]{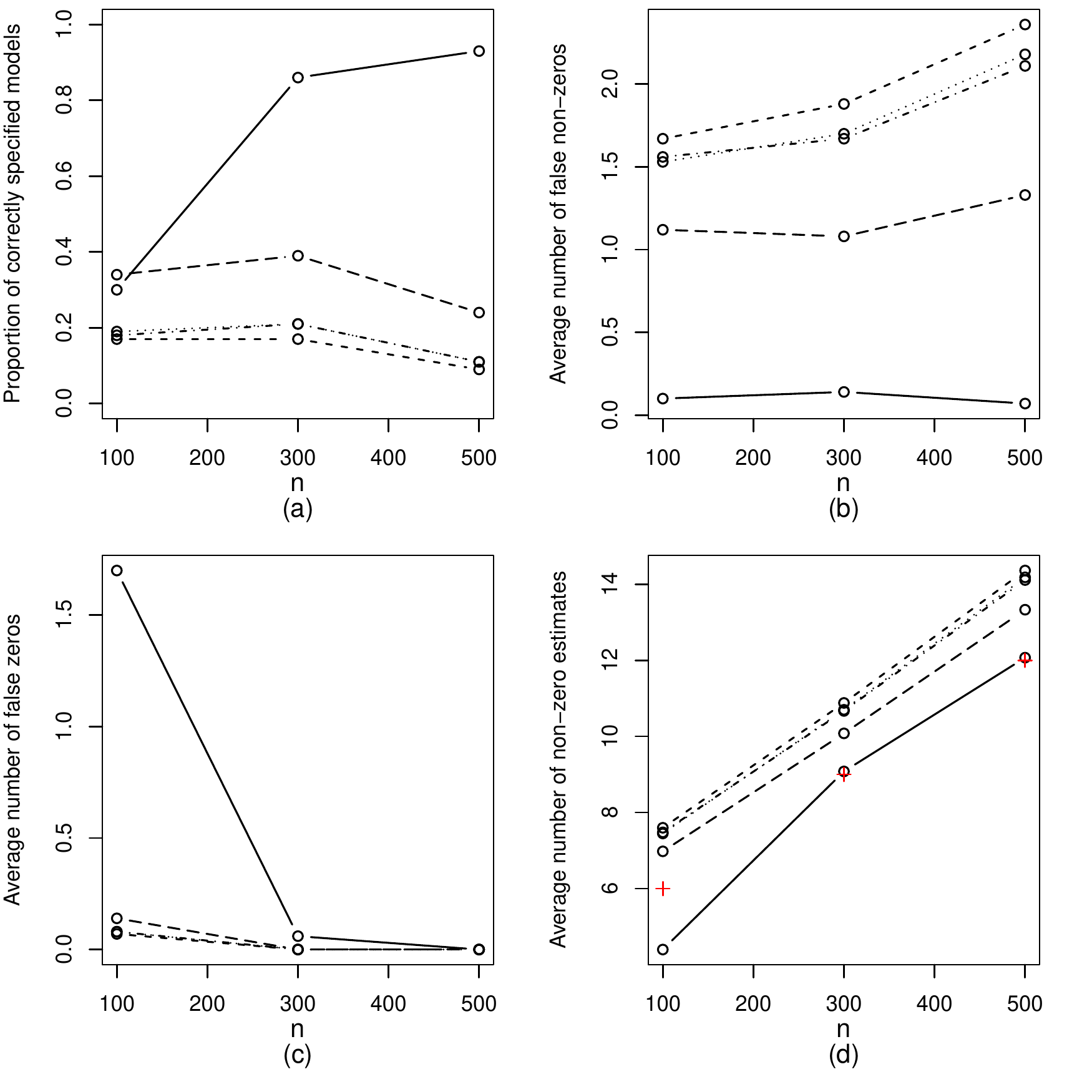}
\caption{Results of scenario 5 using paired bootstrap data: (a) proportion of correctly specified models; (b) average number of false non-zeros; (c) average number of false zeros; (d) average value of estimated model sizes.}
\end{figure}

\begin{figure}
\centering
\includegraphics[scale=0.42]{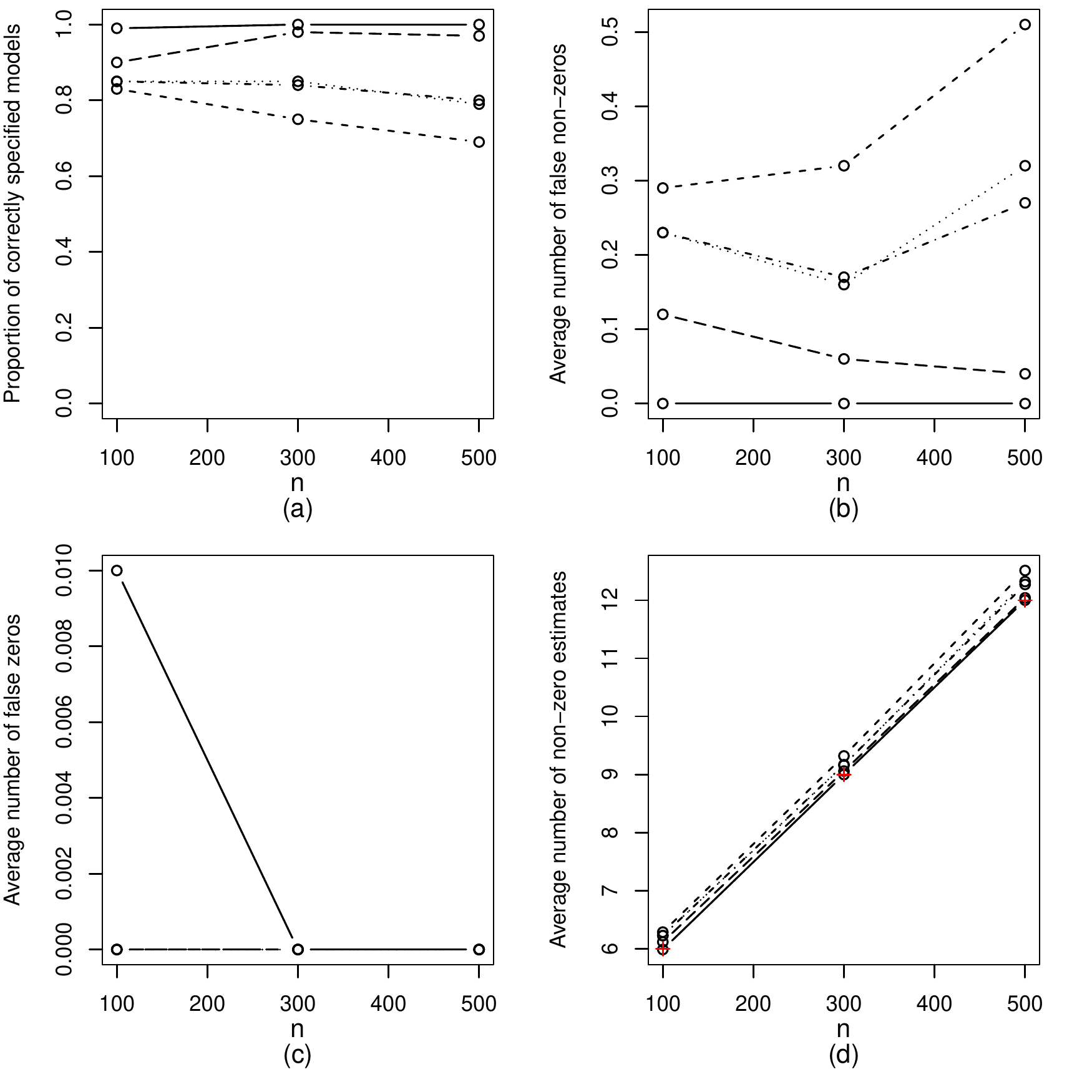}
\caption{Results of scenario 6 using paired bootstrap data: (a) proportion of correctly specified models; (b) average number of false non-zeros; (c) average number of false zeros; (d) average value of estimated model sizes.}
\end{figure}

\newpage
\bibliographystyle{unsrt}
\bibliography{example}

\end{document}